\documentclass[a4paper]{article}
\usepackage[margin=1.12in]{geometry}
\pdfoutput=1 %
\usepackage{microtype}%

\usepackage{nccmath,amsmath,amssymb,mathrsfs}
\usepackage{amsthm}
\numberwithin{equation}{section}  %
\theoremstyle{definition}
\newtheorem{example}{Example}[section]

\newtheorem{definition}{Definition}[section]
\theoremstyle{plain}   %
\newtheorem{theorem}[definition]{Theorem}  %
\newtheorem{proposition}[definition]{Proposition}
\newtheorem{lemma}[definition]{Lemma}

\usepackage{enumerate}
\usepackage{alltt}
\usepackage{color}
\usepackage{graphicx}
\usepackage{url}
\usepackage{latexsym}

\usepackage[labelformat=simple]{subcaption}

\usepackage[all]{xy}
\usepackage{mathtools}
\usepackage{pgf,tikz,pgfplots}
\usetikzlibrary{cd}
\pgfplotsset{compat=1.16}

\newcommand{\complexf}[1]{\mathcal{#1}}   %
\newcommand{\cplC}{\complexf{C}}   %
\newcommand{\cplD}{\complexf{D}}   %
\newcommand{\cplI}{\complexf{I}}   %
\newcommand{\cplO}{\complexf{O}}   %
\newcommand{\cplP}{\complexf{P}} %
\newcommand{\coloring}{\chi}

\newcommand{\PEMCat}{\mathcal{PM}}%
\newcommand{\SimpMCat}{\mathcal{SM}} %

\definecolor{nodewhite}{RGB}{255,255,255}
\definecolor{nodeblue}{RGB}{0,0,204}
\definecolor{nodered}{RGB}{255,204,204}
\newcommand{\nodeW}{\ooalign{\hss{\color{nodewhite}$\bullet$}\hss \crcr $\circ$}}
\newcommand{\nodeB}{\ooalign{\hss{\color{nodeblue}$\bullet$}\hss \crcr $\circ$}}
\newcommand{\nodeR}{\ooalign{\hss{\color{nodered}$\bullet$}\hss \crcr $\circ$}}

\newcommand{\proj}{\pi}  %

\newcommand{\zto}{\to}

\newcommand{\zthen}{\Rightarrow}

\newcommand{\RefClosure}[1]{{#1}^*}

\newcommand{\func}[2]{#1\mathrel{\to}#2}
\newcommand{\anglpair}[1]{\langle{#1}\rangle}

\newcommand{\abs}[1]{{\lvert}#1{\rvert}}
\newcommand{\PowerSet}[1]{2^{#1}}
\newcommand{\ident}[1]{\mathrm{id}_{#1}}

\newcommand{\clos}[1]{{#1}^{\ast}}

\newcommand{\Ag}{\mathsf{Ag}} %
\newcommand{\AtomProps}{\mathsf{AP}}  %
\newcommand{\At}{\mathsf{At}}
\newcommand{\Value}{\mathsf{Value}}

\newcommand{\dec}{\delta} %
\newcommand{\dectop}{\dec_{\mathrm{top}}} %
\newcommand{\deckrip}{\dec_{\mathrm{krip}}} %

\newcommand{\smplMf}[1]{\mathcal{#1}} %

\newcommand{\labSM}{\ell}   %

\newcommand{\precond}{\mathsf{pre}}                  %
\newcommand{\AliveSet}[1]{\overline{#1}}  %

\newcommand{\actMf}[1]{\mathbf{#1}}  %

\newcommand{\actionkappa}{\overline{\kappa}} %
\newcommand{\MPp}[1]{\ActionOne}  %
\newcommand{\MPpA}[1]{T_{#1}} %
\newcommand{\MPprel}[1]{\sim^{T_{#1}}}   %

\newcommand{\MP}[1]{\mathbf{MP}}
\newcommand{\MPA}[1]{T^{\MP{#1}}}
\newcommand{\MPrel}[1]{\sim^{\MP{#1}}}
\newcommand{\MPprecond}[1]{\precond^{\MP{#1}}}

\newcommand{\Inpu}{\mathcal{I}}

\newcommand{\Complex}{\mathcal{C}}

\newcommand{\Mepistemic}{\mathcal{M}}
\newcommand{\Nepistemic}{\mathcal{N}}

\newcommand{\ActionOne}{\mathbf{A}_{0}}
\newcommand{\Action}{\mathbf{A}}

\newcommand{\ActOneSMP}{\mathbf{MP}_{0}}
\newcommand{\ActSMP}{\mathbf{MP}}

\newcommand{\Lo}{\operatorname{lo}}
\newcommand{\lo}{\mathrm{lo}}

\newcommand{\DeltaTop}{\delta_{\mathrm{top}}}
\newcommand{\DeltaKrip}{\delta_{\mathrm{krip}}}

\newcommand{\Facet}{\mathcal{F}}

\newcommand{\Prod}[2]{{#1}\left[{#2}\right]}
\newcommand{\ImProd}[2]{{#1}\left\{{#2}\right\}}
\newcommand{\Abs}[1]{\lvert #1 \rvert}
\newcommand{\Tuple}[1]{\langle #1 \rangle}
\newcommand{\KripkeTuple}[1]{\langle W^{#1}, \sim^{#1}, L^{#1} \rangle}
\newcommand{\ActionTuple}[1]{\langle W^{#1}, \sim^{#1}, \PreOp \rangle}
\newcommand{\SimpTuple}[1]{\langle V^{#1}, S^{#1}, \chi^{#1}, \labSM^{#1} \rangle}

\newcommand{\PreEquiv}[2]{\left[{#2}\right]_{#1}^{\PreOp}}

\newcommand{\SendFailOp}[1]{<_{#1}}

\newcommand{\SendFail}[3]{{#2} \SendFailOp{#1} {#3}}

\newcommand{\NotSendFail}[3]{{#2} \not\SendFailOp{#1} {#3}}
\newcommand{\SameLocalViewForAlivesOp}[1]{\approx_{#1}}

\newcommand{\SameLocalViewForAlives}[3]{{#2} \SameLocalViewForAlivesOp{#1} {#3}}
\newcommand{\ActEquiv}[2]{\left[{#2}\right]_{\SameLocalViewForAlivesOp{#1}}}

\newcommand{\relK}[2]{\sim^{#1}_{#2}}
\newcommand{\clsMark}{\ast}    %
\newcommand{\Rcls}[1]{#1^{\clsMark}}    %
\newcommand{\eqClass}[2]{[#1]_{#2}}  %

\newcommand{\preEqClass}[2]{[#2]_{#1}^{\mathrm{pre}}} %

\newcommand{\satop}{\mathsf{sat}}
\newcommand{\sat}[2]{\satop_{#1} (#2)}
\newcommand{\bigsat}[2]{\satop_{#1} \bigl(#2\bigr)}

\newcommand{\PreOp}{\mathsf{pre}}
\newcommand{\Pre}[1]{\operatorname{\PreOp}{#1}}
\newcommand{\Alives}[1]{\overline{#1}}
\newcommand{\Pfont}[1]{\mathsf{#1}}

\newcommand{\Pfalse}{\Pfont{false}}
\newcommand{\Pinput}[2]{\Pfont{input}_{#1}^{#2}}

\newcommand{\Palive}{\Pfont{alive}}

\newcommand{\Paliveop}[1]{\operatorname{\Palive}({#1})}

\newcommand{\Modalfont}[1]{\mathrm{#1}}
\newcommand{\ModK}[1]{\mathop{\Modalfont{K}_{#1}}} %

\newcommand{\ModLang}{\mathcal{L}}                   %
\newcommand{\ModLangK}{\ModLang_{\Modalfont{K}}} %
\newcommand{\ModLangKPlus}{\ModLang_{\Modalfont{K}, \Palive}^{+}} %

\newcommand{\Axiom}[1]{\mathbf{\mathrm{#1}}}
\newcommand{\AxiomSn}{\Axiom{S5_n}}
\newcommand{\AxiomKBn}{\Axiom{KB4_n}}

\newcommand{\keywd}[1]{\emph{#1}}       %

\title{Partial Product Updates for Agents of Detectable Failure
    and Logical Obstruction to Task Solvability}

\author{
    ~~\hfil%
    \parbox[t]{.25\textwidth}{\centering Daisuke Nakai\hfil\\[1.5ex]}
    ~ \hfil ~
    \parbox[t]{.25\textwidth}{\centering Masaki Muramatsu\\[1.5ex]}%
    ~ \hfil ~
    \parbox[t]{.28\textwidth}{\centering Susumu Nishimura\thanks{Corresponding author: \texttt{susumu@math.kyoto-u.ac.jp}}\\[1.5ex]}%
    \hfil\hfil~~
    \\ ~\\
    Dept.\ Math., Graduate School of Science, Kyoto University
    }

\begin{document}

\maketitle

\begin{abstract}
    The logical method proposed by Goubault, Ledent, and Rajsbaum
    provides a novel way to show  the unsolvability of distributed tasks
    by means of a logical obstruction, which is  an epistemic logic formula
    describing the reason of unsolvability.
    In this paper, we introduce the notion of partial product update, which refines
    that of product update in the original logical method, 
    to encompass distributed tasks and protocols modeled
    by impure simplicial complexes. 
    With this extended notion of partial product update, the original
    logical method  is generalized so that it allows the application of
    logical obstruction to show unsolvability results in a distributed environment
    where the failure of agents is detectable.
    We demonstrate the use of the logical method 
    by giving a concrete logical obstruction and showing that the consensus
    task is unsolvable by the single-round synchronous message-passing protocol.
\end{abstract}

\section{Introduction}
\label{sec:intro}

The solvability of distributed tasks is a fundamental problem in the theory of 
distributed computing, and 
several proof methods for showing unsolvability have been developed.  
The most classical is based on the valency argument \cite{FischerLynchPaterson85,Herlihy91}: 
If we assume the solvability of a task, 
it implies a contradictory set of outputs to be produced along concurrent execution paths.
Another method is the one that uses the topological model \cite{Book:2013:HerlihyKozlovRajsbaum}:
In the topological method, the computation of a distributed task or a
protocol is modeled by a function over simplicial complexes representing
the nondeterministic sets of states of concurrently running agents, and
the unsolvability of a task is demonstrated by a topological
inconsistency that is implied by a hypothetical existence of a
solution to the task.

In addition to these precursors,
a new method, called the \keywd{logical method}, has recently been proposed by 
Goubault, Ledent, and Rajsbaum \cite{Inf:2021:GoubaultLedentRajsbaum}. 
They proposed to discuss the structure of distributed computation 
in a Kripke model of epistemic knowledge, which is derived from the topological model
of simplicial complex. They formulated the distributed task and protocol using the so-called 
product update models, which originate from the study of dynamic epistemic logic \cite{DitmarschHoekKooi:DELbook08}.
The logical method provides a novel way to show the unsolvability of distributed tasks:
The unsolvability follows from a \keywd{logical obstruction}, i.e., 
a formula that describes the reason for the unsolvability in the formal
language of epistemic logic. 

Their logical method, however, only applies
to the tasks and protocols where the possible failure of distributed agents is insignificant
for the discussion of unsolvability.
The semantics of epistemic logic is given by the so-called \keywd{epsistemic models}, 
where an epistemic model is a Kripke model whose possible worlds are structured 
by equivalence relations over them.
They also showed that the epistemic models used in the logical method
and the pure simplicial complexes used in the topological method are isomorphic models. 
The purity means that each facet is of the same dimension, that is, 
each possible global state of a distributed system consists of 
the same number of live agents. 

This implies that the epistemic models in \cite{Inf:2021:GoubaultLedentRajsbaum} 
are unaware of failures, since we need impure simplicial complexes to model 
`dead' agents that are missing from a facet.
Despite the innocence of failure, the epistemic models can 
argue certain significant unsolvability results, including those for
the consensus task and $k$-set agreement tasks by a wait-free protocol 
in an asynchronous environment, 
where a `dead' agent can be regarded as just infinitely slow in execution.
(See Section~1.1.3 of \cite{PhD:2019:Ledent} for further discussion.)

Later in \cite{STACS22:GoubaultLedentRajsbaum}, the same authors devised 
\keywd{partial epistemic models}, whose possible worlds are structured 
by partial equivalence relations (PERs), 
which are subordinate equivalence relations that are not necessarily reflexive,
and argued that partial epistemic models are appropriate for 
modeling the possible failure of agents.
Partial epistemic models inherit several virtues from epistemic models.
They are isomorphic to the topological model of impure simplicial complexes;
They provide the semantics for the epistemic logic with axiom system~$\AxiomKBn$, 
as so do epistemic models for axiom system~$\AxiomSn$ \cite{Book:2004:JosephMoses,DitmarschHoekKooi:DELbook08};
Furthermore, they enjoy the knowledge gain property, 
which states that the knowledge expressed by an adequate class of formulas 
never increases along a morphism over the models.

However, the framework of logical method has not been presented for partial epistemic models in~\cite{STACS22:GoubaultLedentRajsbaum}. 
The notion of product update relevant to partial epistemic models is needed for the definition of task solvability, but they did not present it.

In this paper, we introduce the notion of \keywd{partial product update}, 
which refines the original product update \cite{Inf:2021:GoubaultLedentRajsbaum}.
Using partial product update models, we give a logical definition of 
task solvability and thereby provide the logical method for proving task unsolvability.
Furthermore we present a concrete logical obstruction 
and show that the consensus task is not solvable by 
the synchronous message passing protocol \cite{S1571:2001:HerlihyRajsbaumTuttle}. 

The partial product update model refines the original model to 
allow coherent definitions of distributed tasks and protocols
in a distributed environment in which the exact set of dead agents is detectable. 
In a product update model $\Prod{\cplI}{\actMf{A}}$ of \cite{Inf:2021:GoubaultLedentRajsbaum}, 
a task or a protocol is specified by a set of products $(X,t)$ where an action $t$ in 
$\actMf{A}$ represents a possible output for the input $X$ in the input model~$\cplI$. 
In contrast, in partial epistemic models, an action $t$ 
should be associated with not necessarily a single input 
but with a set of inputs that are indistinguishable by the agents that are alive in $t$. 
For this, we define a partial product update model $\ImProd{\cplI}{\actMf{A}}$
as a set of products of the form $(\eqClass{X}{t},t)$, where $\eqClass{X}{t}$ is
an appropriate equivalence class of $X$ determined with respect to 
the set of agents that are alive in $t$. 

Using partial product update models, we define task solvability by 
the existence of a morphism that mediates the partial product update model of a protocol and that of a task.
We will show that this definition is 
equivalent to the topological definition of task solvability
in the following sense: 
Given a task and a protocol as functions over simplicial complexes,
we derive an action model and a partial product update model for each of them. 
Then a simplicial map that defines the solvability in the topological model 
exists if and only if a mediating morphism 
over partial epistemic models exists. 
This definition of task solvability using partial product updates 
refines the one using product updates so that it allows a detectable set of failed agents.
Furthermore, likewise in \cite{Inf:2021:GoubaultLedentRajsbaum}, 
it provides the logical method that allows a logical obstruction 
to show the unsolvability of a task.

We demonstrate the use of logical obstruction in partial product update models
by showing that the consensus task is unsolvable by the single-round
synchronous message passing protocol. 
In the synchronous message passing protocol \cite{S1571:2001:HerlihyRajsbaumTuttle},
each agent sends copies of its local value to other agents synchronously. 
Unlike wait-free protocols in an asynchronous environment, 
a crash of an agent in a synchronous environment is detectable by other live agents, 
since the failure of message delivery from a dead agent can be detected within a bounded period of time. 
This gives rise to an impure simplicial complex model, as depicted in Figure~\ref{fig:impureSMP3}.
In order to argue unsolvability for such an impure simplicial complex in
partial epistemic model, we construct a partial product update from an action model 
whose actions are represented by posets of rank at most $1$. 
We present a concrete epistemic logic formula that serves as a logical 
obstruction that refutes the solvability of the consensus task.

\begin{figure}[t]
    \centering
    \includegraphics[scale=0.5]{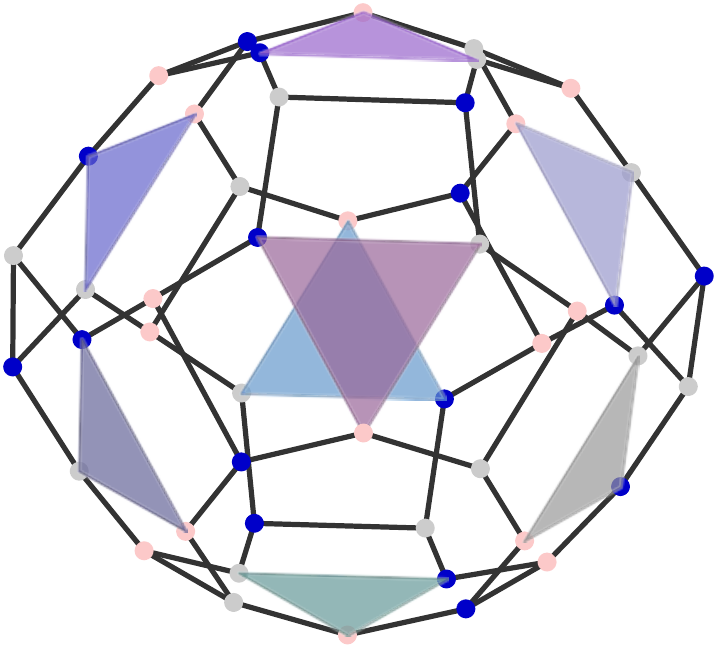}
    \caption{An impure simplicial complex of the synchronous message passing protocol for~$3$ agents,
    with $0$-dimensional facets being omitted}
    \label{fig:impureSMP3}
\end{figure}

We note that 
the previous work \cite{S1571:2001:HerlihyRajsbaumTuttle}
analyzes more precisely the (un)solvability of $k$-set agreement tasks
by the $r$-round synchronous message passing protocol, for varying $k$ and $r$,  
using the topological method.
In the present paper, we only demonstrate the unsolvability of 
the consensus task (i.e., the $1$-set agreement task) by the single-round 
synchronous message passing protocol. 
To show the other unsolvability results, we would need to further refine
the logical method presented in this paper,  as it has been done for the
logical method in asynchronous environments,   e.g.,
\cite{Nishida:Msc20,YagiNishimura:arXiv20}.

The rest of the paper is organized into the following sections.
Section~\ref{sec:partialEpiModel} reviews the two previous models of distributed computing, 
namely, the topological models and the epistemic models.
In Section~\ref{sec:PPU}, we introduce the notion of partial product update and
give the definition of task solvability in partial epistemic models.
We also show that this definition of task solvability is 
equivalent to the standard one given in terms of the topological method.
(The formal proof of the equivalence is given in Appendix~\ref{sec:solvablEq}.)
Section~\ref{sec:SMPactionmodel} defines the partial product update model for the 
synchronous message passing protocol, where the action model consists of actions 
that are represented by posets of rank at most $1$. 
Section~\ref{sec:obstructionSMP} gives a concrete logical obstruction 
and proves that the consensus task is not solvable by the single-round 
synchronous message passing protocol.
Finally Section~\ref{sec:conclusion} concludes the paper 
and indicates the direction of future research.

\section{Partial Epistemic Model for Epistemic Logic}
\label{sec:partialEpiModel}

Throughout this paper, we assume a distributed system consisting of $n$ ($n>1$) agents, 
which are distinguished by the identifiers taken from the set $\Ag = \{0, \ldots, n-1 \}$.
We say `agent~$a$' to refer to the agent with id $a \in \Ag$. 

\subsection{Topological model of distributed computing}
\label{subsec:topModelDistComp}

In the topological method \cite{Book:2013:HerlihyKozlovRajsbaum}, 
distributed systems are modeled by simplicial complexes. 
A global state of a distributed system is modeled 
by a \keywd{simplex}, i.e., a nonempty set of vertexes, where each vertex
represents the local state of a particular agent.
The nonderteministic set of global states of a distributed system is 
modeled by a \keywd{simplicial complex}, a set of simplexes that is closed 
under set inclusion. 
The topological model uses the so called \emph{chromatic}
simplicial complex, whose vertexes are properly `colored' with agent ids. 

\begin{definition}[chromatic simplicial complex]
    A \keywd{chromatic simplicial complex} $\cplC$ is a triple $\Tuple{V, S, \coloring}$ consisting
    of a set $V$ of vertexes, a set $S$ of \keywd{simplexes}, 
    and a \keywd{coloring map} $\coloring: V \zto \Ag$ that satisfy:
    \begin{itemize}
        \item $S$ is the set of simplexes, i.e., a set of nonempty subsets of $V$ 
        such that $X\in S$ and $\emptyset\subsetneq Y\subseteq X$ implies $Y \in S$;
        \item For every $X \in S$ and $u,v\in X$, $\coloring(u)=\coloring(v)$ implies $u=v$.
    \end{itemize}

\end{definition}

For brevity, we often write complexes (resp., simplexes) to mean chromatic simplicial complexes (resp., chromatic simplexes).

A simplex $X$ is of dimension $d$, if $\Abs{X} =d+1$.
The dimension of a complex $\cplC$ is the maximum dimension of the simplexes contained in $\cplC$. 
A simplex $X$ is called a \keywd{facet} of $\cplC$, if 
$X$ is a maximal simplex, i.e., $\cplC$ contains no facet that is properly larger than $X$. 
We write $\Facet(\cplC)$ for the set of facets of $\cplC$.
A complex is called \keywd{pure} (of dimension $n-1$) if all facets are of the same dimension~$n-1$;
otherwise, the complex is called \keywd{impure}.
An impure complex contains a facet of dimension less than $n-1$. 
Such a facet represents a global state of the distributed system 
where some agents are `dead' due to crash.

Given complexes $\Complex = \Tuple{V, S, \coloring}$ and $\Complex' = \Tuple{V', S', \coloring'}$,
a \keywd{simplicial map} $\delta:\cplC\to\cplC'$ is a color-preserving
map from $V$ to $V'$ such that:
\begin{itemize}
    \item $f(X) \in S'$ for every $X\in S$;
    \item $\coloring'(f(v)) = \coloring(v)$ for every $v \in V$.
\end{itemize}

\subsection{Partial epistemic model semantics for epistemic logic}
\label{subsec:partialEpistemicModel}

We are concerned with the analysis of distributed computability 
with epistemic logic \cite{Book:2004:JosephMoses}, a propositional logic augmented with a knowledge modality. 

\begin{definition}[The syntax of epistemic logic] \label{def:SyntaxEpistemicLogic}
    We assume the set $\At$ is a disjoint union of 
    atomic propositions indexed by agents, i.e., $\At =\bigcup_{a\in\Ag} \At_a$. 
    For $A \subseteq \Ag$, we write $\At_A$ for the set
    $\bigcup_{a \in A} \At_a$ of atomic propositions concerning the subset $A$ of agents.

    The set $\ModLangK$ of epistemic logic formulas are defined by the following BNF grammar:
    \[
        \varphi ::=
        p
        \mid \neg \varphi
        \mid \varphi \wedge \varphi \mid \varphi \vee \varphi
        \mid \ModK{a}{\varphi}
        \qquad (p \in \At,\: a \in \Ag).
    \]

    As usual, the implication $\varphi \zthen \varphi'$ is logically equivalent to $\neg \varphi \vee \varphi'$
    and $\Pfalse$ to $p\wedge \neg p$, for some $p\in\At$.  
    For a finite set of formulas $G = \{\varphi_1, \cdots, \varphi_k\}$, we write 
    $\bigwedge G$ for $\bigwedge_{i=1}^k \varphi_i$ 
    and $\bigvee G$ for $\bigvee_{i=1}^k \varphi_i$.

    For $a\in \Ag$ and $B\subseteq \Ag$, 
    we write $\Paliveop{a}$ to abbreviate the formula $\neg \ModK{a} \Pfalse$ 
    and also $\Paliveop{B}$ to abbreviate $\bigwedge_{a \in B} \Paliveop{a}$. 
    The formula $\Paliveop{a}$ (resp., $\Paliveop{B}$) is intended to 
    describe that agent $a$ (resp., all agents in $B$) are alive. 

    In the subsequent discussion, we are concerned with a particular subclass $\ModLangKPlus$ 
    of formulas, called \keywd{guarded positive epistemic formulas}, given by:
    \[
        \varphi ::=
        \Paliveop{B} \zthen \psi
        \mid \varphi \wedge \varphi
        \mid \varphi \vee \varphi
        \mid \ModK{a}{\varphi}
        \qquad (a \in \Ag,~ B \subseteq \Ag),
    \]
    where $\psi$ ranges over pure propositional formulas whose atomic formulas 
    are restricted to those concerning agents in $B$, namely, 
    \[
        \psi ::=
        p
        \mid \neg \psi
        \mid \psi \wedge \psi
        \mid \psi \vee \psi
        \qquad (p \in \At_B).
    \]
\end{definition}

Following \cite{STACS22:GoubaultLedentRajsbaum}, 
below we define the semantics of epistemic logic in a partial epistemic model.

\begin{definition}[partial epistemic model] \label{def:PEframe}\label{def:PEmodel}
    A \keywd{partial epistemic frame} $\Tuple{W, \sim}$ is a pair consisting of:
    \begin{itemize}
        \item A nonempty finite set $W$ of possible worlds;
        \item An \keywd{indistinguishability relation} $\sim$, 
        which is a family of binary relations $\{ {\sim_a} \}_{a\in\Ag}$ over $W$  
        where each $\sim_a$ is a \keywd{partial equivalence relation} (PER), i.e., 
        $\sim_a$ is a symmetric and transitive, but not necessarily reflexive relation.
    \end{itemize}

    A \keywd{partial epistemic model} $\Mepistemic=\Tuple{W, \sim, L}$ 
    is a triple 
    that augments a partial epistemic frame $\Tuple{W, \sim}$ with a function $L: W \to \PowerSet{\At}$.
    The function $L$ assigns a set $L(w)$ of true atomic propositions to each world $w \in W$.
\end{definition}

When $w \sim_a w'$ holds in a partial epistemic model,     
it means that agent $a$ is alive in
both of the possible worlds $w$ and $w'$ and it cannot distinguish between these worlds.
Particularly, $w \not\sim_a w$ implies that
an agent~$a$ is dead in a possible world $w$.
We write $\Alives{w}$ for the set of agents that is alive in $w$, 
i.e., $\Alives{w} = \{ a \in \Ag \mid w \sim_a w \}$. 
For $A\subseteq\Ag$, we also write $w \sim_A w'$ to mean
$w \sim_a w'$ for every $a\in A$. 

Partial epistemic models generalize the usual epistemic models: 
Indistinguishability relations are given by equivalence relations instead of PERs,
which means that every agent is alive in epistemic models.

\begin{definition}
    Let $\Mepistemic = \Tuple{W, \sim, L }$ be
    a partial epistemic model.
    Given $w \in \Mepistemic$
    and $\varphi \in \ModLangK$,
    the satisfaction relation $\Mepistemic, w \models \varphi$,  
    which reads $\varphi$ is true in the possible world $w$ of $\Mepistemic$, 
    is defined by induction on the structure of $\varphi$ as follows.
    \begin{align*}
        \Mepistemic, w \models p
         &\quad \text{ iff } \quad p \in L(W)       \\
        \Mepistemic, w \models \neg \varphi
         &\quad \text{ iff } \quad \Mepistemic, w \not\models \varphi                                      \\
         \Mepistemic, w \models \varphi \wedge \varphi'
         &\quad \text{ iff } \quad \Mepistemic, w \models \varphi \text{ and } \Mepistemic, w \models \varphi' \\
         \Mepistemic, w \models \varphi \vee \varphi'
         &\quad \text{ iff } \quad \Mepistemic, w \models \varphi \text{ or } \Mepistemic, w \models \varphi' \\
        \Mepistemic, w \models \ModK{a}{\varphi}
         &\quad \text{ iff } \quad \Mepistemic, w' \models \varphi
        \text{ for every $w' \in W$ satisfying $w \sim_a w'$}.
    \end{align*}

    We write $\Mepistemic \models \varphi$ to mean $\varphi$ is \keywd{valid} in $\Mepistemic$, 
    i.e., $\Mepistemic, w \models \varphi$ for every $w\in W$. 
\end{definition}

Partial epistemic models form a category 
whose objects are partial epistemic models and morphisms are the functions defined as follows.
A \keywd{morphism} $f$ from $\Mepistemic = \Tuple{W, \sim, L}$ to $\Mepistemic' = \Tuple{W', \sim', L'}$
is a function $f: W \to \PowerSet{W'}$ satisfying the following properties:
\begin{itemize}
    \item (Preservation of $\sim$) For all $a\in\Ag$ and $w, w' \in M$, 
    $w \sim_a w'$ implies $u\sim_a' u'$ for every $u\in f(w)$ and $u'\in f(w')$. 
    \item (Saturation) For all $w\in W$, there exists $w' \in f(w)$ such that 
    $f(w) = \sat{\Alives{w}}{w'}$, where $\sat{U}{v}=
    \{ w \in W \mid  v \sim_U w \}$ 
    is a \keywd{saturation}, the set of all possible worlds that cannot be 
    distinguished from $w$ by any agent $a\in U$.

    \item (Preservation of atomic formulas)
    For all $w \in M$ and $w'\in f(w)$, $L(w) \cap \At_{\Alives{w}}= L'(w)\cap \At_{\Alives{w}}$.
\end{itemize}

In order to show unsolvability results in this paper, 
we will use the \keywd{knowledge gain} property, 
which states that the amount of knowledge is never increased along a morphism 
from one Kripke model to another. 
For partial epistemic models, the knowledge gain property holds 
for any guarded positive formula \cite{STACS22:GoubaultLedentRajsbaum}.

\begin{proposition}[knowledge gain]
    \label{prop:knowledgeGainLk}
    Let $\Mepistemic = \KripkeTuple{\Mepistemic}$
    and $\Nepistemic = \KripkeTuple{\Nepistemic}$ be partial epistemic models.
    Let $f: \Mepistemic \to \Nepistemic$ be
    a morphism such that $w' \in f(w)$
    for all $w \in \Mepistemic$ and $w' \in \Nepistemic$.
    Then
    $\Nepistemic, w' \models \varphi$ implies $\Mepistemic, w \models \varphi$
    for any guarded positive epistemic formula $\varphi \in \ModLangKPlus$.
\end{proposition}

\subsection{Simplicial models and the equivalence with partial epistemic models}
\label{subsec:simplicialModel}
\label{subsec:catEquivalence}

From the topological structure of a given (not necessarily pure) simplicial complex,
we can derive \keywd{simplicial models} \cite{STACS22:GoubaultLedentRajsbaum}.
\begin{definition}[simplicial models]
    A \keywd{simplicial model} $\Complex$ is a quadruple 
    $\Tuple{V, S, \coloring, \labSM}$ consisting of:
    \begin{itemize}
        \item The triple $\Tuple{V, S, \coloring}$ of the underlying complex;
        \item The \keywd{labeling} $\labSM$ 
        that assigns a set $\labSM(X)$ of atomic propositions to each facet $X$ of the complex.
    \end{itemize}

\end{definition}
The set $\labSM(X)$ of atomic propositions 
determines the local states of agents in a given global state represented by $X$.

From a simplicial model $\Complex = \Tuple{V, S, \coloring, \labSM}$, 
we can derive a partial epistemic model $\anglpair{\Facet(\cplC),\sim,L}$  such that
\begin{itemize}
    \item The set of possible worlds is $\Facet(\Complex)$, the set of facets of 
    the complex $\Complex$;
    \item The indistinguishability relation is defined by 
    $X \sim_a Y$ if and only if $a\in \coloring (X \cap Y)$. That means, 
    the facets $X$ and $Y$ 
    sharing a common vertex of color~$a$ are indistinguishable by the agent~$a$;
    \item The labeling $L$ on possible worlds is defined by 
    $L(X)= \labSM(X)$ 
    for each $X\in\Facet(\Complex)$.
\end{itemize}

The partial epistemic model derived from a simplicial model in this way 
is a proper Kripke model.  
A Kripke model $\Mepistemic=\Tuple{W,\sim,L}$
is called \keywd{proper} if
$w \neq w'$ implies $\exists a \in \Ag.~w \not\sim_a w'$ for all $w, w' \in W$.
It has been shown that 
the category $\PEMCat$ of proper partial epistemic models 
is equivalent to the category $\SimpMCat$ of simplicial models 
\cite{STACS22:GoubaultLedentRajsbaum}, 
where a morphism $f$ from $\Complex = \Tuple{V, S, \coloring, \labSM}$ 
to $\Complex' = \Tuple{V', S', \coloring', \labSM'}$ 
in $\SimpMCat$ is a simplicial map 
$f: \Complex \to \Complex'$ satisfying 
$\labSM(X)\cap\At_{\coloring(X)} = \labSM'(Y)\cap\At_{\coloring(X)}$ 
for every $X\in \Facet(\Complex)$ and $Y\in\Facet(\Complex')$ 
such that $f(X)\subseteq Y$. 

Due to this equivalence, in the sequel
we will occasionally confuse simplicial models with proper partial epistemic models.
We regard facets of a complex $\cplC$, say $X, Y \in \Facet(\cplC)$, 
as possible worlds of the corresponding partial epistemic model
and argue the indistinguishability by the property of the facets.
For example, 
$X \sim_a Y$ in a partial equivalence model is interpreted as 
$a \in \coloring(X\cap Y)$;
The notation~$\AliveSet{X}$, which denotes the set of live agents 
in a facet $X$ of a simplicial model, is given by $\AliveSet{X}=\coloring(X)$.  

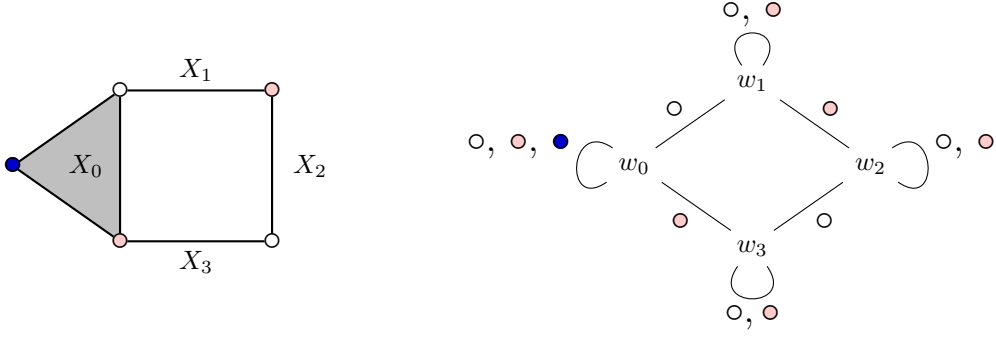
\begin{figure}[t]
    \hspace*{.05\textwidth}%
    \begin{minipage}{.28\textwidth}
        \begin{tikzpicture}
            \draw[thick] (-1, -1) -- (1, -1) -- (1, 1) -- (-1, 1);
            \draw[fill=lightgray, thick] (-1, 1) -- (-1, -1) -- (-2.41, 0) -- cycle;
            \draw (1.5,0) node {$X_{2}$};
            \draw (0,1) node[anchor=south] {$X_{1}$};
            \draw (0,-1) node[anchor=north] {$X_{3}$};
            \draw (-1.1,0) node[anchor=east] {$X_{0}$};
            \draw (-2.41,0) node {\Large $\nodeB$};
            \draw (-1, 1) node {\Large $\nodeW$};
            \draw (1, 1) node {\Large $\nodeR$};
            \draw (1, -1) node {\Large $\nodeW$};
            \draw (-1, -1) node {\Large $\nodeR$};
        \end{tikzpicture}
    \end{minipage}%
    ~\hfil\hfil~ %
    \begin{minipage}{.5\textwidth}
        \begin{tikzcd}
            & w_{1} \arrow[rd, "\text{\Large $\nodeR$}", no head] \arrow["\text{\Large \nodeW, \nodeR}"', no head, loop, distance=2em, in=125, out=55] &          \\
            w_{0} \arrow[ru, "\text{\Large \nodeW}", no head] \arrow[rd, "\text{\Large \nodeR}"', no head] \arrow["\text{\Large \nodeW, \nodeR, \nodeB}"', no head, loop, distance=2em, in=215, out=145] &    & w_{2} \arrow[ld, "\text{\Large \nodeW}", no head] \arrow["\text{\Large \nodeW, \nodeR}"', no head, loop, distance=2em, in=35, out=325] \\
            & w_{3} \arrow["\text{\Large $\nodeW, \nodeR$}"', no head, loop, distance=2em, in=305, out=235]                                     &
        \end{tikzcd}
    \end{minipage} 

    \caption{An impure simplicial model and its corresponding partial epistemic model}
    \label{fig:twoModels}
\end{figure}

\begin{example}
    Figure~\ref{fig:twoModels} depicts a simplicial model $\anglpair{V,S,\coloring,\labSM}$
    (on the left)
    and its corresponding partial epistemic model $\anglpair{M,\sim,L}$ (on the right)
    for $3$ agents with colors $\Ag=\{0,1,2\}$. 
    In the figure, agents $0$, $1$, and $2$ are represented by 
    colored nodes $\nodeW$, $\nodeR$, and $\nodeB$, respectively.  
    Throughout the paper, we will follow this coloring convention.

    For each $i=0,1,2,3$, 
    the facet $X_i$ in the simplicial model corresponds to 
    the possible world~$w_i$ in the partial epistemic model 
    and they have the same set of atomic propositions, i.e., $\ell(X_i) = L(w_i)$. 
    If a pair of facets $X_i$ and $X_j$ share a node of color~$a$, 
    the corresponding possible worlds $w_i$ and $w_j$ are connected by an edge associated 
    with a node of color~$a$, meaning that $w_i \sim_a w_j$.
    Particularly, each possible world $w_i$ has a self-loop colored by $a$ 
    if and only if the agent~$a$ is alive in $w_i$, that is, 
    the corresponding facet $X_i$ contains a vertex colored by~$a$.

\end{example}

\section{Partial Product Update and Task Solvability} \label{sec:PPU}

In \cite{Inf:2021:GoubaultLedentRajsbaum}, 
product update is used to define distributed tasks and protocols. 
A product update model $\Prod{\cplI}{\actMf{A}}$ 
is an epistemic model where each possible world
is a pair $(X,t)$ of a facet $X$ of the complex $\cplI$
and an action $t$ of the action model $\actMf{A}$ 
such that the action $t$ stands for a possible output for 
the input $X$. Task solvability is then 
defined by the existence of a morphism from 
the product update model of the protocol to that of the task. 

We introduce a generalized construction called \keywd{partial product update}, 
which applies to partial epistemic models, and thereby 
refine the logical method so that it can be applied to a larger class of
task solvability problem in which agents may die.

\subsection{Partial product update}
\label{subsec:PPUbyAction}

\begin{definition}[action model]
    An \keywd{action model} $\Action=\langle T, \sim, \PreOp \rangle$ is
    a triple consisting of: 
    \begin{itemize}
        \item A partial epistemic frame $\langle T, \sim \rangle$, 
        where each $t\in T$ is called an \keywd{action};
        \item A function $\PreOp: T \zto \ModLangK$ that assigns
        an epistemic formula $\PreOp(t)$, called the \keywd{precondition},
        to each action $t\in T$.  
    \end{itemize}

\end{definition}

The precondition $\PreOp(t)$ describes the condition 
for the action $t$ to occur.

\begin{definition}[partial product update]
    \label{def:partialProductUpdate}
    Let $\Mepistemic = \KripkeTuple{\Mepistemic}$ be a partial epistemic model
    and $\Action = \ActionTuple{\Action}$ be an action model,
    both having a proper partial epistemic frame. 
    A \keywd{partial product update} of $\Mepistemic$ by $\Action$, 
    written $\ImProd{\Mepistemic}{\Action}$, 
    is a partial epistemic model $\KripkeTuple{}$ defined by:
    \begin{itemize}
        \item $W = \{
                  (\PreEquiv{t}{X}, t)
                  \mid
                  X \in W^{\Mepistemic},\:
                  t \in W^{\Action},\:
                  \AliveSet{t}\subseteq \AliveSet{X}, \text{ and }
                  \Mepistemic, X \models \Pre{t}
                  \}$, 

        \item $
                  (\PreEquiv{t}{X}, t) \sim_a (\PreEquiv{s}{Y}, s)$ iff
                  $X \sim^{\Mepistemic}_a Y$ and $t \sim^{\Action}_a s
              $, and 

        \item $
                  L\bigl((\PreEquiv{t}{X}, t)\bigr) =
                  \displaystyle\bigcap_{X' \in \PreEquiv{t}{X}} L^{\Mepistemic}(X')
              $, 
    \end{itemize}
    where $\PreEquiv{t}{X}$ is the set of facets defined by 
    $\PreEquiv{t}{X} = \{
        Y \in W^{\Mepistemic}
        \mid
        X \sim^{\Mepistemic}_{\Alives{t}} Y
        \text{ and }
        \Mepistemic, Y \models \Pre{t}
        \}$.

\end{definition}

\begin{proposition}
    The partial product update 
    $\ImProd{\Mepistemic}{\Action}$ is a partial epistemic model.
    In particular, if both  $\Mepistemic$ and $\Action$ are proper, 
    so is $\ImProd{\Mepistemic}{\Action}$.
\end{proposition}

\begin{proof}  
    Let $\ImProd{\Mepistemic}{\Action}=\KripkeTuple{}$ be a triple, 
    which is constructed from a partial epistemic model
    $\Mepistemic = \KripkeTuple{\Mepistemic}$ 
    and an action model $\Action = \ActionTuple{\Action}$, 
    as in Definition~\ref{def:partialProductUpdate}.

    Let us first show that the relation $\sim_a$ is well-defined,
    that is, the choice of facet $X$ in $(\PreEquiv{t}{X}, t)$ 
    does not affect the relation.  
    Suppose $(\PreEquiv{t}{X}, t) \sim_a (\PreEquiv{s}{Y}, s)$,
    $\PreEquiv{t}{X} = \PreEquiv{t}{X'}$, and
    $\PreEquiv{s}{Y} = \PreEquiv{s}{Y'}$. 
    By the definition, $X \sim^{\Mepistemic}_a Y$, 
    $\Mepistemic, X\models \PreOp(t)$, $\Mepistemic, X'\models \PreOp(t)$, 
    $\Mepistemic, Y\models \PreOp(s)$, $\Mepistemic, Y'\models \PreOp(s)$,    
    and $t \sim^{\Action}_a s$.
    Since $X' \in \PreEquiv{t}{X'} = \PreEquiv{t}{X}$, we have 
    $X \sim^{\Mepistemic}_a X'$ for all $a\in\Alives{t}$. Similarly, 
    $Y \sim^{\Mepistemic}_a Y'$ for all $a\in\Alives{t}$. 
    Hence $X' \sim^{\Mepistemic}_a X \sim^{\Mepistemic}_a Y \sim^{\Mepistemic}_a Y'$
    and therefore $(\PreEquiv{t}{X'}, t) \sim_a (\PreEquiv{s}{Y'}, s)$.

    We show that $\ImProd{\Mepistemic}{\Action}$ is a partial epistemic model.
    It suffices to show that every $\sim_a$ is a PER. 
    This follows immediately from the symmetry and transitivity 
    of $\sim^{\Mepistemic}_a$ and $\sim^{\Action}_a$.

    Suppose $\Mepistemic$ and $\Action$ are both proper models. 
    We show that  $\ImProd{\Mepistemic}{\Action}=\KripkeTuple{}$ is a proper model.
    Assume $(\PreEquiv{t}{X}, t), (\PreEquiv{s}{Y}, s) \in W$
    are possible worlds such that $(\PreEquiv{t}{X}, t) \neq (\PreEquiv{s}{Y}, s)$, 
    which implies $\PreEquiv{t}{X} \neq \PreEquiv{s}{Y}$ or $t \neq s$. 
    It suffices to show that there exists an agent $a \in \Ag$
    such that $(\PreEquiv{t}{X}, t) \not\sim_a (\PreEquiv{s}{Y}, s)$.
    If $t \neq s$, since $\Action$ is proper, 
    $t \not\sim^{\Action}_a s$ for some $a \in \Ag$. Thus, 
    $(\PreEquiv{t}{X}, t) \not\sim_a (\PreEquiv{s}{Y}, s)$.
    Suppose otherwise, i.e., $\PreEquiv{t}{X} \neq \PreEquiv{s}{Y}$.
    We may assume, without loss of generality, that 
    $X' \not\in \PreEquiv{t}{Y}$ for some $X' \in \PreEquiv{t}{X}$.
    This implies $\PreEquiv{t}{X} = \PreEquiv{t}{X'}$ and $\Mepistemic, X' \models \Pre{t}$.
    Hence, $X' \not\in \PreEquiv{t}{Y}$ implies
    $X' \not\sim^{\Mepistemic}_a Y$ for some agent $a \in \Alives{t}$.
    Therefore we have $(\PreEquiv{t}{X}, t) = (\PreEquiv{t}{X'}, t) \not\sim_a (\PreEquiv{t}{Y}, t) = (\PreEquiv{s}{Y}, s)$.
\end{proof}

The original product update $\Prod{\Mepistemic}{\actMf{A}}$ \cite{Inf:2021:GoubaultLedentRajsbaum}
is an instance of the partial product update $\ImProd{\Mepistemic}{\actMf{A}}$.
When the indistinguishability relations of both $\Mepistemic$ and $\actMf{A}$ are equivalence relations, 
rather than partial equivalence relations, and $\Mepistemic$ is proper, 
every possible world $(\PreEquiv{t}{X},t)$ in $\ImProd{\Mepistemic}{\actMf{A}}$ 
can be identified with $(X,t)$ in $\Prod{\Mepistemic}{\actMf{A}}$.
This is because $\PreEquiv{t}{X}=\{ X \}$, since 
$\Alives{t} = \Ag$ for all $t \in \Action$ and $\Mepistemic$ is proper.

\subsection{Task solvability in partial product update models}
\label{subsec:tasksolv}

In the sequel, we assume that the set $\At$ of atomic propositions 
is given by $\At = \bigcup_{a\in\Ag} \At_a$ with
$\At_a=\{\Pinput{a}{v}\mid v\in\Value\}$ for each $a\in\Ag$,
where $\Value$ is the set of possible input values.  
Without loss of generality, we may assume that 
the set of input values coincides with the set of agent ids, i.e.,
$\Value=\Ag=\{0,\cdots, n-1\}$.

The input model for a system with $n$ agents is given as follows.

\begin{definition}[input simplicial model]\label{def:inputsimplicialmodel}
    An \keywd{input simplicial model} $\Inpu = \Tuple{V, S, \chi, \labSM}$ consists of:
    \begin{itemize}
        \item The complex $\Tuple{V, S, \chi}$ consisting of 
        the set of vertexes $V = \{ (a, v) \mid a \in \Ag,$ $v \in \Value \}$, 
        the coloring map defined by $\chi\bigl((a, v)\bigr) = a$, and the set of 
        simplexes $S = \{
            X \mid
            \emptyset \subsetneq X \subseteq \{
            (0, v_{0}), \ldots, (n-1, v_{n-1})
            \}$ where 
            $v_{0}, \ldots, v_{n-1} \in \Value \}$;
        \item The labeling defined by $\labSM(X) = \{ \Pinput{a}{v} \mid (a,v)\in X\}$
        for each facet $X$ of the complex.
    \end{itemize}
\end{definition}

In the rest of the paper, we write $\cplI$ to refer to this particular input simplicial model. 
The underlying simplicial complex $\Tuple{V, S, \chi}$ of 
$\Inpu$ is a pure complex of dimension~$n-1$, where
each vertex $(a,v)$ represents a local state of 
the agent~$a$ that has $v$ as the input value. 
The set of facets in the complex represents all the possible assignments of input values 
to the agents, 
and the labeling $\labSM$ interprets such assignments by 
giving the relevant set of atomic propositions for each facet $X$, 
so that $\Pinput{a}{v}\in \labSM(X)$ if and only if $(a,v)\in X$.

For example, Figure~\ref{fig:inputcpl_n2} illustrates the underlying complex of 
the input simplicial model $\Inpu$ for $2$~agents (i.e, $\Ag=\{0,1\}$).
In the figure, we write $X_{ij}$ to stand for a facet $\{(0,i),(1,j)\}$ of~$\cplI$. 
The labeling is determined by $\labSM(X_{ij})= \{ \Pinput{0}{i}, \Pinput{1}{j} \}$.

\begin{figure}[ht]
    \centering
    \begin{tikzpicture}[scale=0.88]
        \draw[thick](1,1)--(1,-1)--(-1,-1)--(-1,1)--cycle;

        \draw(-1,1)node{{\Large $\nodeW$}};
        \draw(-1,-1)node{{\Large $\nodeR$}};
        \draw(1,1)node{{\Large $\nodeR$}};
        \draw(1,-1)node{{\Large $\nodeW$}};

        \draw(-1.2,1.3)node{$(0,0)$};
        \draw(1.2,1.3)node{$(1,0)$};
        \draw(-1.2,-1.3)node{$(1,1)$};
        \draw(1.2,-1.3)node{$(0,1)$};

        \draw(0,1.3)node{$X_{00}$};
        \draw(-1.5,0)node{$X_{01}$};
        \draw(1.5,0)node{$X_{10}$};
        \draw(0,-1.3)node{$X_{11}$};
    \end{tikzpicture}
    \caption{The input simplicial model $\cplI$ for $n=2$}
    \label{fig:inputcpl_n2}
\end{figure}
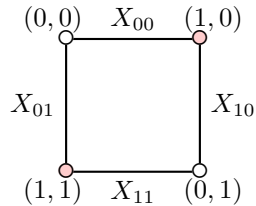

Now we define task solvability, using partial product update. 

\begin{definition}    \label{def:kripkesolvability}
Let $\cplI$ be the input simplicial model and let 
$\actMf{P}$ and $\actMf{T}$ be action models for a protocol and a task, respectively.
A task $\actMf{T}$ is \keywd{solvable} by the protocol $\actMf{P}$ 
if there exists a morphism $\DeltaKrip: \ImProd{\cplI}{\actMf{P}} \to \ImProd{\cplI}{\actMf{T}}$ 
such that %
\begin{equation} \label{eq:kripsolve}
    \forall (\preEqClass{p}{X_p},p) \in \ImProd{\cplI}{\actMf{P}}.\: 
    \exists (\preEqClass{t}{X_t}, t)\in\deckrip\bigl((\preEqClass{p}{X_p}, p)\bigr). \:
    \preEqClass{p}{X_p}\subseteq \preEqClass{t}{X_t}.
\end{equation}
\end{definition}

Schematically, this definition of solvability can be presented by the following 
diagram-like picture:
\[ \xymatrix{
    \cplI \ar @{} [d] |{\rotatebox{270}{$\subseteq$}} & \ImProd{\cplI}{\actMf{P}} \ar[l]_{\proj_1}\ar[d]^-{\exists \DeltaKrip}\\
    \cplI    & \ImProd{\cplI}{\actMf{T}} \ar[l]^{\proj_1} 
    } \]
where $\proj_1$ is the first projection.

Suppose a morphism $\DeltaKrip$ associates
$(\preEqClass{p}{X_p}, p)\in \cplI\{\actMf{P}\}$ with
$(\preEqClass{t}{X_t}, t)\in \cplI\{\actMf{T}\}$. 
Then, in order for this association to be admissible, 
it must respect the inclusion of~\eqref{eq:kripsolve},
meaning that,  when $\DeltaKrip$ decides an output $t$ from 
an intermediate output $p$ by the protocol, 
every input contained in 
$\preEqClass{p}{X_p}$, i.e., the set of inputs from which the protocol may produce the output $p$, 
must be an input contained in $\preEqClass{t}{X_t}$, i.e., the set of inputs 
that admit $t$ as the output of the task.

This generalizes the original definition of task solvability given in \cite{Inf:2021:GoubaultLedentRajsbaum},
which makes use of (non-partial) product updates. 
The original definition is given by the following commutative diagram, which is obtained by
replacing the models by the product update models, written 
$\cplI[\actMf{P}]$ and $\cplI[\actMf{T}]$, and 
the inclusion in~\eqref{eq:kripsolve} by equality $=$.
\[ \xymatrix{
    & \Prod{\cplI}{\actMf{P}} \ar[dl]_{\proj_1}\ar[d]^-{\exists \delta} \\
    \cplI    & \Prod{\cplI}{\actMf{T}} \ar[l]^{\proj_1} 
    } \]

This definition of the task solvability using product update models is 
equivalent to the topological definition using simplicial complexes \cite{Inf:2021:GoubaultLedentRajsbaum}.  
We can also establish the same equivalence for partial product update models, in the following sense:
Suppose we are given a pair of functions over simplicial complexes 
specifying a task and a protocol. From these specifications we can derive 
the corresponding partial product update models.
Then the task is solvable by the protocol, using partial product update models 
as defined in Definition~\ref{def:kripkesolvability},  
if and only if the task is solvable in the topological sense, i.e., there exists a 
simplicial map that mediates between the complexes of the protocol and the task.

The formal argument for this equivalence is deferred to Appendix~\ref{sec:solvablEq}. 
There  we also show that 
the action model $\actMf{A}$ and 
the partial product update model $\ImProd{\cplI}{\actMf{A}}$ 
have isomorphic partial epistemic frames, under a suitable condition. 

 \subsection{Logical obstruction theorem}
 \label{subsec:logicalObstruction}

Combining the above task solvability with the knowledge gain property (Proposition~\ref{prop:knowledgeGainLk}), 
we obtain a proof method for refuting task solvability by a formula of epistemic logic, 
called a \keywd{logical obstruction}~\cite{Inf:2021:GoubaultLedentRajsbaum}.
A logical obstruction $\varphi$ is a guarded positive formula that is valid in the 
partial product update model $\ImProd{\cplI}{\actMf{T}}$ of a task, but not in the
model $\ImProd{\cplI}{\actMf{P}}$ of a protocol. 

\begin{theorem}\label{thm:LOTheorem}
    Let $\ImProd{\cplI}{\actMf{P}}$ and $\ImProd{\cplI}{\actMf{T}}$ be partial product update models
    of a protocol and a task, respectively. 
    If there exists a guarded positive formula $\varphi \in \ModLang_{\Modalfont{K}, \Palive}^{+}$  
    such that $\ImProd{\cplI}{\actMf{T}} \models \varphi$ but $\ImProd{\cplI}{\actMf{P}} \not\models \varphi$, 
    the task is not solvable by the protocol. 
\end{theorem}
\begin{proof}
    Suppose that $\varphi$ is a guarded positive formula such that 
    $\ImProd{\cplI}{\actMf{T}} \models \varphi$ but $\ImProd{\cplI}{\actMf{P}} \not\models \varphi$
    and, by contradiction, that the task is solvable, i.e., there exists 
    a morphism $\delta$ that satisfies the condition~\eqref{eq:kripsolve}. 
    The invalidity for the protocol $\ImProd{\cplI}{\actMf{P}} \not\models \varphi$ 
    means that  $\varphi$ is false for some world $w$ of $\ImProd{\cplI}{\actMf{P}}$.
    Then Proposition~\ref{prop:knowledgeGainLk} implies that there exists a world $w'$ of 
    $\ImProd{\cplI}{\actMf{T}}$ such that $\ImProd{\cplI}{\actMf{T}}, w'\not\models \varphi$. 
    This contradicts to the validity of $\varphi$ in $\ImProd{\cplI}{\actMf{T}}$. 
\end{proof}

\section{Action Model for Synchronous Message Passing Protocol}
\label{sec:SMPactionmodel}

\subsection{Actions as posets of rank at most $1$} \label{subsec:actionAsPosets}

In the synchronous message passing protocol,
each agent sends its local private value to all the agents in the system.
When an agent crashes,
only the messages that have already been sent before the crash
are properly delivered to other agents; the rest are not.
In a synchronous distributed environment,
a correct agent can detect the exact set of crashed agents in the system:
An agent is dead (i.e., it has crashed), if no message is received from
the agent within a sufficiently long but finite period of time.
The protocol complex for the synchronous message passing protocol thus
forms an impure simplicial complex \cite{S1571:2001:HerlihyRajsbaumTuttle},
in which the vertexes corresponding to dead agents are missing in a facet.

In the impure protocol complex, 
each different facet arises from each different `pattern of failure' in
message delivery, where each specific pattern of failure
can be described by a set of orderings of the form $a < b$,
which means that the agent~$a$ has crashed and failed to send the message to the agent~$b$, 
a correct agent that did not crash.
This motivates us to define actions by relevant orderings over the set of agents
that satisfy the following properties.
\begin{itemize}
    \item An agent $a$ is dead after an action (i.e., the agent has crashed during the execution of the protocol)
          if and only if the corresponding poset of the action admits an ordering $a<b$ for some agent~$b$.
    \item An agent $a$ is alive after an action
          if and only if the corresponding poset of the action does not admit
          an ordering $a<b$ for any agent~$b$.
    \item If a poset contains an ordering $a<b$,
          the poset does not admit an ordering $c<a$ or $b<c$ for any agent~$c$.
          (This is because the ordering $a<b$ indicates that $a$ is a dead agent and
          $b$ is a correct agent but $c<a$ or $b<c$ implies otherwise.)
\end{itemize}

This ordering can be formalized as a poset (partially ordered set)
$(\Ag, \leq)$ of rank at most $1$.
A poset is of \keywd{rank}~$k$,
if the maximum number of elements contained
in a totally ordered subset of the poset is $k+1$. \cite{Book:2011:Stanley}
As usual, the strict ordering $a<b$ holds if and only if $a\leq b$ and $a\neq b$.

Let us define $\Lo(\leq) = \{ a \in \Ag \mid \exists b \in \Ag, a < b \}$
to denote the set of dead agents.
Notice that the reflexive ordering $a\leq a$ is not counted as a failure of message delivery.
We are not interested in the success or failure of self-delivery of message:
If successful, the agent~$a$ trivially receives its own local private value;
Otherwise, the agent~$a$ is dead and its received messages are lost.
In this respect, we indicate a poset by a set $S$ of strict orderings
and write $\Rcls{S}$ to denote the minimum poset containing $S$, i.e., the reflexive closure of $S$.
For example, $\Rcls{\{0<1, 2<3\}}$ defines a poset ordering $\leq$ over $\Ag$ such that
$a\leq b$ holds if and only if either $a=0$ and $b=1$,  $a=2$ and $b=3$, or $a=b$.
In particular, $\Rcls{\emptyset}$ is a discrete poset (of rank $0$) over $\Ag$,
in which any pair of distinct agents are incomparable.

\subsection{Constructing the action model $\ActSMP$} 
\label{subsec:constActModel}

Let us first consider the action model for the `inputless' case,
where the initial inputs to the agents are insignificant.
In this case, the actions for the synchronous message passing protocol is simply modeled 
by the posets over $\Ag$ of rank at most $1$.

\begin{definition} \label{def:smp_one}
    The \keywd{inputless action model} for synchronous message passing protocol is
    a partial epistemic frame 
    $\ActOneSMP = \langle T, \sim\rangle$ consisting of:
    \begin{itemize}
        \item The set of actions $
            T = \{
            t \mid t \text{ is a poset of rank at most $1$ over $\Ag$} \}$ and

        \item The family of indistinguishability relation, where $t_1 \sim_a t_2$ is defined 
        for each $a\in\Ag$ by: 
        \[
            t_1 \sim_a t_2 ~\text{ iff }~
            a \notin \Lo(t_1) \cup \Lo(t_2) \text{ and }
            \{ b \in \Ag \mid \SendFail{t_1}{b}{a} \} = \{ b \in \Ag \mid \SendFail{t_2}{b}{a} \}.
            \]
    \end{itemize}

    For each $t\in T$, we denote the set of live agents by 
    $\Alives{t} = \Ag \setminus \Lo(t)= \{ a \in \Ag \mid t \sim_a t \} $.    
\end{definition}

\begin{figure}[t]
    \begin{center}
        \begin{tikzpicture}[scale=1.41]
            \draw[fill=lightgray, thick](0,1)--(0.87,-0.5)--(-0.87,-0.5)--cycle;
            \draw(0,0)node{\small $\RefClosure{\emptyset}$};

            \draw[thick](0,1)--(1.5,1.87)--(2.37,0.37)--(0.87,-0.5);
            \draw(0.75,1.8)node{\small $\RefClosure{\{(\nodeW<\nodeR)\}}$};
            \draw(3,1.2)node{\small $\RefClosure{\{(\nodeW<\nodeB),(\nodeW<\nodeR)\}}$};
            \draw(2.05,-0.3)node{\small $\RefClosure{\{(\nodeW<\nodeB)\}}$};

            \draw[thick](0,1)--(-1.5,1.87)--(-2.37,0.37)--(-0.87,-0.5);
            \draw(-0.75,1.8)node{\small ~ $\RefClosure{\{(\nodeR<\nodeW)\}}$};
            \draw(-3,1.2)node{\small $\RefClosure{\{(\nodeR<\nodeB),(\nodeR<\nodeW)\}}$};
            \draw(-2.05,-0.3)node{\small $\RefClosure{\{(\nodeR<\nodeB)\}}$};

            \draw[thick](0.87,-0.5)--(0.87,-2.23)--(-0.87,-2.23)--(-0.87,-0.5);
            \draw(1.5,-1.37)node{\small $\RefClosure{\{(\nodeB<\nodeW)\}}$};
            \draw(-1.5,-1.37)node{\small $\RefClosure{\{(\nodeB<\nodeR)\}}$};
            \draw(0,-2.5)node{\small $\RefClosure{\{(\nodeB<\nodeW),(\nodeB<\nodeR)\}}$};

            \draw(0,1.0)node{\Large $\nodeB$};
            \draw(1.5,1.87)node{\Large $\nodeR$};
            \draw(2.37,0.37)node{\Large $\nodeB$};
            \draw(0.87,-0.5)node{\Large $\nodeR$};
            \draw(0.87,-2.23)node{\Large $\nodeW$};
            \draw(-0.87,-2.23)node{\Large $\nodeR$};
            \draw(-0.87,-0.5)node{\Large $\nodeW$};
            \draw(-2.37,0.37)node{\Large $\nodeB$};
            \draw(-1.5,1.87)node{\Large $\nodeW$};

            \draw(0,3.34)node{\Large $\nodeB$};
            \draw(0,3.0)node{\small $\RefClosure{\{(\nodeR<\nodeB),(\nodeW<\nodeB)\}}$};

            \draw(-2.9,-1.67)node{\Large $\nodeW$};
            \draw(-2.5,-1.95)node[anchor=east]{\small $\RefClosure{\{(\nodeB<\nodeW),(\nodeR<\nodeW)\}}$};

            \draw(2.9,-1.67)node{\Large $\nodeR$};
            \draw(2.5,-1.95)node[anchor=west]{\small $\RefClosure{\{(\nodeB<\nodeR),(\nodeW<\nodeR)\}}$};
        \end{tikzpicture}
    \end{center}
    \caption{The inputless action model $\ActOneSMP$ of the synchronous message passing protocol for $3$~agents, 
    with facets annotated by the corresponding posets over $\Ag$}
    \label{fig:3_protKripSmpOne}
\end{figure}
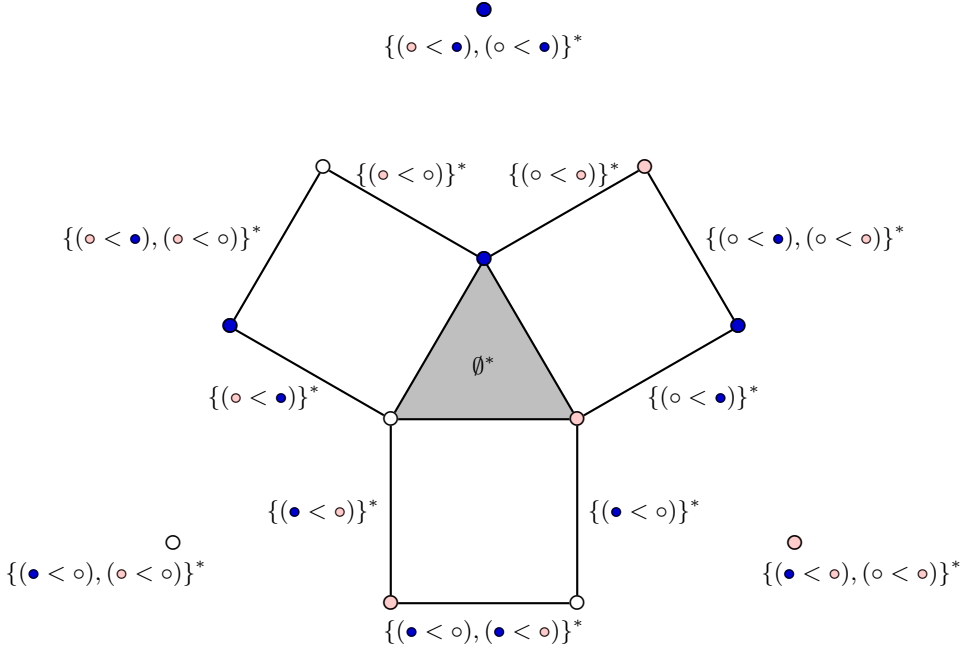

\begin{example}
    Figure~\ref{fig:3_protKripSmpOne} illustrates the complex of 
    the inputless action model $\ActOneSMP$ for $3$~agents. 
    The complex is impure, where the facet corresponding to an action $t$
    is of dimension $\abs{\AliveSet{t}}-1$. %
    When $t \sim_a^{\ActOneSMP} s$ holds, 
    the condition $a \notin \Lo(t) \cup \Lo(s)$ means that  
    the agent $a$ is alive for both actions $t$ and $s$; 
    The condition $\{ b \in \Ag \mid \SendFail{t}{b}{a} \} = \{ b \in \Ag \mid \SendFail{s}{b}{a} \}$
    means that $t$ and $s$ have the same set of agents that failed to send a message to~$a$
    and hence they know exactly the same set of delivered messages. 
    For example, let 
    $t_1=\RefClosure{\{(\nodeB<\nodeR)\}}$, 
    $t_2=\RefClosure{\{(\nodeB<\nodeW), (\nodeB<\nodeR)\}}$, 
    $t_3=\RefClosure{\{(\nodeB<\nodeW)\}}$.  
    Then $t_1 \sim_{\nodeR}^{\ActOneSMP} t_2$ holds because, 
    in both $t_1$ and $t_2$, the agent~$\nodeR$ is alive 
    and has successfully received messages from 
    the agents other than $\nodeB$. 
    On the other hand, $t_1 \sim_{\nodeR}^{\ActOneSMP} t_3$ does not hold
    because the agent~$\nodeW$ in $t_3$ has received the messages from all
    the agents, but the same agent~$\nodeW$ in $t_1$ did not, from the
    agent $\nodeB$. 
\end{example}

From the inputless action model, 
we construct a full action model of the synchronous message passing protocol,
where the protocol accepts a class of sets of input values. 
Suppose $\Inpu$ is a simplicial model, whose facets specify the allowed combinations of input values. 
Then one might expect that an action in the full action model 
would be expressed by a pair $(X, t)$, where $X$ is a facet of $\Inpu$ 
and $t$ is an action of the inputless action,
but this does not give an appropriate model.
Consider the case, say, $t = \RefClosure{\{ b<a \}}$, 
$X \sim_{a}^{\Inpu} Y$, and $X \not\sim_{b}^{\Inpu} Y$. 
In this case, $(X, t)$ and $(Y, t)$ are different actions but the live
agent $a$ should not distinguish them because the action $t$ indicates
that the agent~$a$ has not received a message from the agent~$b$ and
therefore it cannot tell which is the original input, either $X$ or $Y$.
In other words, the two facets $X$ and $Y$ should be identified. 

For this identification of facets w.r.t.\ an action~$t$ of the inputless model,  
we introduce an equivalence class, written $\ActEquiv{t}{X}$,
and define an action in the full action model by a pair $(\ActEquiv{t}{X},t)$.

\begin{definition}[action model of synchronous message passing]
    \label{def:SMP}    \label{def:sameLocalViewForAlives}
    Let $\Inpu = \SimpTuple{\Inpu}$ be an input simplicial model.
    For each facet $X, Y \in \Facet(\Inpu)$ and action $t \in \ActOneSMP$,
    we define an equivalence relation $\SameLocalViewForAlives{t}{X}{Y}$ by:
    \[
        \SameLocalViewForAlives{t}{X}{Y}
        \text{ iff }
        \forall a \in \Alives{t}.\:
        \forall b \in \Ag.\: (
            X \sim^{\Inpu}_b Y \vee \SendFail{t}{b}{a}
            ).
        \]
        We write $\ActEquiv{t}{X}$ to denote the equivalence class of $X$ w.r.t.\ 
        $\SameLocalViewForAlivesOp{t}$. 
    
    We define
    the action model $\ActSMP=\anglpair{T^{\ActSMP}, \sim^{\ActSMP}, \PreOp^{\ActSMP}}$ of
    the synchronous message passing protocol as a triple consisting of:
    \begin{itemize}
        \item The set of actions 
        $T^{\ActSMP} = \{ (\ActEquiv{t}{X}, t) \mid X \in \Facet(\Inpu),~ t \in \ActOneSMP \}$;
        \item The family of indistinguishability relations defined by 
        \[
            (\ActEquiv{t}{X}, t) \sim^{\ActSMP}_a (\ActEquiv{s}{Y}, s)
              ~\text{ iff }~
              t \sim^{\ActOneSMP}_a s
              \text{ and }
              \forall b \in \Ag.\: (X \sim^{\Inpu}_b Y \vee \SendFail{t}{b}{a});
          \]
        \item The precondition defined by
        $
        \PreOp^{\ActSMP} \bigl((\ActEquiv{t}{X}, t)\bigr)
                  = \bigvee \{
                  \bigwedge \labSM^{\Inpu}(Y) \mid Y \in \ActEquiv{t}{X}
                  \}
              $. %
    \end{itemize}
\end{definition}

In words, $\SameLocalViewForAlives{t}{X}{Y}$ indicates that,
for every agent $a$ that is alive in the action $t$, 
$X$ and $Y$ should have the same input for the agent~$b$, if 
$\NotSendFail{t}{b}{a}$, i.e., the message from~$b$ is successfully 
received by $a$. (Otherwise, the agent~$a$ should be aware that 
$X$ and $Y$ are different facets.)
The relation $\SameLocalViewForAlivesOp{t}$ is an equivalence relation, 
as we will show below. 
The relation $\sim^{\ActSMP}_a$ defines a PER, where 
$(\ActEquiv{t}{X}, t) \sim^{\ActSMP}_a (\ActEquiv{s}{Y}, s)$
indicates that the agent~$a$ is alive
and every agent~$b$ that successfully sends a message to~$a$
(i.e., $\NotSendFail{t}{b}{a}$) has the same input value that has been assigned 
in both $X$ and $Y$.  
The precondition $\PreOp^{\ActSMP} \bigl((\ActEquiv{t}{X}, t)\bigr)$ is 
intended to mean that, for an action $(\ActEquiv{t}{X}, t)\in T^{\ActSMP}$, 
every atomic proposition in $\labSM^{\Inpu}(Y)$ is true 
for some $Y \in \ActEquiv{t}{X}$.

\begin{proposition}
    \label{prop:actSimWellDefined}
    Let $\Inpu$ be an input simplicial model. 
    For every $t\in \ActOneSMP$, $\SameLocalViewForAlivesOp{t}$ is an equivalence relation. 
    Furthermore, the relation $\sim^{\ActSMP}_a$ is a PER and is well-defined. 
    That is, $\sim^{\ActSMP}_a$ is a symmetric transitive relation and is not affected by
    the choice of facet $X$ %
    in $(\ActEquiv{t}{X},t)$.
\end{proposition}

\begin{proof}
    Obviously, $\SameLocalViewForAlivesOp{t}$ is a reflexive symmetric relation. 
    For transitivity, assume $\SameLocalViewForAlives{t}{X}{Y}$ and $\SameLocalViewForAlives{t}{Y}{Z}$.
    Suppose $X \not\sim^{\Inpu}_b Z$. By the transitivity of $\sim^{\Inpu}_b$, 
    we have $X \not\sim^{\Inpu}_b Y$ or $Y \not\sim^{\Inpu}_b Z$. 
    In both cases, we have $\SendFail{t}{b}{a}$ for all $a\in\Ag$. 
    Hence $\SameLocalViewForAlives{t}{X}{Z}$. 

    To show that $\sim^{\ActSMP}_a$ is symmetric, suppose 
    $(\ActEquiv{t}{X}, t) \sim^{\ActSMP}_a (\ActEquiv{s}{Y}, s)$. 
    By the definition, $t \sim^{\ActOneSMP}_a s$ and 
    $\forall b \in \Ag.~(X \sim^{\Inpu}_b Y \vee \SendFail{t}{b}{a})$.
    Since $t \sim^{\ActOneSMP}_a s$ implies 
    $\{ b \in \Ag \mid \SendFail{t}{b}{a} \} = \{ b \in \Ag \mid \SendFail{s}{b}{a} \}$, 
    by the symmetry of $\sim^{\ActOneSMP}_a$ and $\sim^{\Inpu}_b$, 
    we have $(\ActEquiv{s}{Y}, s) \sim^{\ActSMP}_a (\ActEquiv{t}{X}, t)$. 
    The transitivity of $\sim^{\ActSMP}_a$ follows similarly from
    the transitivity of $\sim^{\ActOneSMP}_a$ and $\sim^{\Inpu}_b$.

    For the well-definedness of $\sim^{\ActSMP}_a$, 
    suppose $(\ActEquiv{t}{X}, t) \sim^{\ActSMP}_a (\ActEquiv{s}{Y}, s)$, 
    $\ActEquiv{t}{X}=\ActEquiv{t}{X'}$, and 
    $\ActEquiv{s}{Y}=\ActEquiv{s}{Y'}$. We show
    $(\ActEquiv{t}{X'}, t) \sim^{\ActSMP}_a (\ActEquiv{s}{Y'}, s)$. 
    Since $\sim^{\ActOneSMP}_a$ is a PER, it suffices to show 
    that $\forall b \in \Ag.~(X' \sim^{\Inpu}_b Y' \vee \SendFail{t}{b}{a})$.
    Assume otherwise, i.e., $X' \not\sim^{\Inpu}_b Y'$ and $\NotSendFail{t}{b}{a}$
    holds for some $b\in\Ag$. Then, from
    $(\ActEquiv{t}{X}, t) \sim^{\ActSMP}_a (\ActEquiv{s}{Y}, s)$, 
    we have $a \in \Alives{t}$ and $a \in \Alives{s}$
    and also entail $X \sim^{\Inpu}_b Y$ from $\NotSendFail{t}{b}{a}$. 
    Since $\SameLocalViewForAlives{t}{X}{X'}$, we have $X \sim^{\Inpu}_b X'$.
    Similarly, $Y \sim^{\Inpu}_b Y'$. By the transitivity of 
    $\sim^{\Inpu}_b$, we obtain $X' \sim^{\Inpu}_b Y'$, a contradiction. 
\end{proof}

    \begin{figure}[t]
        \centering
                \begin{tikzpicture}[scale=0.5]
                    \draw[fill=lightgray, thick](2,1)--(2,-1)--(3.73,0)--cycle;

                    \draw[fill=lightgray, thick](-2,1)--(-2,-1)--(-3.73,0)--cycle;

                    \draw[thick](2,1)--(-2,1);

                    \draw[thick](2,-1)--(-2,-1);
                    \draw[thick](0,1)--(0,-1);

                    \draw[thick](2,1)--(3,2.73)--(4.73,1.73)--(3.73,0);
                    \draw[thick](2,-1)--(3,-2.73)--(4.73,-1.73)--(3.73,0);

                    \draw[thick](-2,1)--(-3,2.73)--(-4.73,1.73)--(-3.73,0);
                    \draw[thick](-2,-1)--(-3,-2.73)--(-4.73,-1.73)--(-3.73,0);

                    \draw[very thick, red] (0, 1) -- (0, -1);
                    \draw[very thick, blue] (-2, 1) -- (2, 1);
                    \draw[thick] (-2, -1) -- (2, -1);

                    \draw (0, 2.73) node {\Large $\nodeR$};
                    \draw (0, -2.73) node {\Large $\nodeW$};
                    \draw (-6.0, 0) node {\Large $\nodeB$};
                    \draw (6.0, 0) node {\Large $\nodeB$};

                    \draw(0,1)node{\Large $\nodeW$};
                    \draw(0,-1)node{\Large $\nodeR$};

                    \draw(2,1)node{\Large $\nodeR$};
                    \draw(2,-1)node{\Large $\nodeW$};
                    \draw(3.73,0)node{\Large $\nodeB$};
                    \draw(3,2.73)node{\Large $\nodeB$};
                    \draw(4.73,1.73)node{\Large $\nodeR$};
                    \draw(3,-2.73)node{\Large $\nodeB$};
                    \draw(4.73,-1.73)node{\Large $\nodeW$};

                    \draw(-2,1)node{\Large $\nodeR$};
                    \draw(-2,-1)node{\Large $\nodeW$};
                    \draw(-3.73,0)node{\Large $\nodeB$};
                    \draw(-3,2.73)node{\Large $\nodeB$};
                    \draw(-4.73,1.73)node{\Large $\nodeR$};
                    \draw(-3,-2.73)node{\Large $\nodeB$};
                    \draw(-4.73,-1.73)node{\Large $\nodeW$};
                \end{tikzpicture}

        \caption{An action model $\ActSMP$ for $3$ agents, 
                generated from an input model consisting of two facets}
        \label{fig:3_smp_two_simplicial}
    \end{figure}
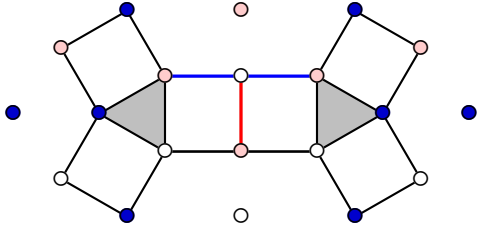

\begin{example} \label{eg:smpTwo}
    Let us consider an input simplicial model $\cplI$ for $3$~agents,  
    consisting of  two facets $X = \{ (\nodeW, 0), (\nodeR, 0), (\nodeB, 0) \}$ and 
    $Y = \{ (\nodeW, 0), (\nodeR, 0), (\nodeB, 1) \}$.
    Then we obtain the action model $\ActSMP$ as depicted in Figure~\ref{fig:3_smp_two_simplicial}. 
    We obtain this complex by first making copies of the complex of the inputless action model,
    for each facet $X$ and $Y$, and pasting them appropriately. 
    For example, the action represented by the poset $t=\RefClosure{\{\nodeB<\nodeW, \nodeB<\nodeR\}}$ 
    generates two copies $(\ActEquiv{t}{X}, t)$ and $(\ActEquiv{t}{Y}, t)$, but 
    these are identified (as depicted by the $1$-dimensional simplex marked by the red line in the figure)
    because $\ActEquiv{t}{X} = \ActEquiv{t}{Y} = \{X, Y\}$.
    Similarly, the topmost node $\nodeR$ and the bottom node~$\nodeW$ in the figure
    identify the corresponding $0$-dimensional simplexes in the copies generated from $X$ and $Y$. 

    Consider another action $s=\RefClosure{\{\nodeB<\nodeW, \nodeB<\nodeR\}}$.
    This generates two copies $(\ActEquiv{s}{X}, s)$ and $(\ActEquiv{s}{Y}, s)$, 
    which are depicted in the figure by the two $1$-dimensional simplexes marked by the blue lines.
    These simplexes are connected via a common vertex $\nodeW$. That is, 
    $(\ActEquiv{s}{X}, s) \sim_{\nodeW}^{\ActSMP} (\ActEquiv{s}{Y}, s)$ holds,
    since $X \sim_{\nodeR}^{\Inpu} Y$ and $\nodeR$ is the only agent $b$
    that satisfies $\NotSendFail{s}{b}{\nodeW}$.
    The two simplexes are not identified nevertheless, since 
    $(\ActEquiv{s}{X}, s) \not\sim_{\nodeR}^{\ActSMP} (\ActEquiv{s}{Y}, s)$ 
    follows from $X \not\sim_{\nodeR}^{\Inpu} Y$ and 
    $\NotSendFail{s}{\nodeR}{\nodeW}$.

\end{example}

\section{Logical Obstruction for Synchronous Message Passing Protocol}
\label{sec:obstructionSMP}

In this section, we argue the solvability of the consensus task by the
single-round synchronous message passing protocol, using partial product
update models.

In \cite{S1571:2001:HerlihyRajsbaumTuttle}, a precise topological analysis
has been given on the solvability of $k$-set agreement tasks by
synchronous message passing protocols of varying protocol rounds.
We demonstrate that the techniques developed in the previous sections can
be applied to reproduce a part of these results: The consensus task is
solvable by a single-round execution of the synchronous message passing
protocol if the number $n$ of agents in the system is $2$ or less;
otherwise, when $n>2$, it is unsolvable.

\subsection{Action model for the consensus task}
\label{subsec:consensustask}

The consensus task is specified by the following 
constraints on the input and output values of the agents.
\begin{itemize}
    \item \textbf{Agreement}:
          Every non-faulty agent must decide on the same output value;
    \item \textbf{Validity}:
          The agreed output value must be an input value to some of the agents in the system.
\end{itemize}

Without loss of generality, we may assume the set of possible inputs 
for the consensus task is $\{0,\ldots, n-1\}$, the set of agent ids. 

\begin{definition}[Partial product update model of the consensus task]\label{def:consensus}
    The consensus task is defined by a partial product update model $\ImProd{\cplI}{\actMf{T}}$,
    where  $\cplI$ is the input simplicial model and  
    $\actMf{T}=\anglpair{T, \relK{T}{}, \precond^T}$ is the action model consisting of:
\begin{itemize}
    \item The set of actions $T= \Ag = \{0,\ldots, n-1\}$;
    \item The family of indistinguishability relations $\{ \relK{T}{a} \}_{a\in\Ag}$
          defined by  $v \relK{T}{a} v'$ iff $v=v'$;
    \item The precondition defined by 
          $\precond^T (v) =\bigvee_{a\in \Ag} \Pinput{a}{v}$
          for each action $v\in T$.
\end{itemize}
\end{definition}

An action in $v\in T$ stands for a single output value unanimously decided by all live agents.
Each agent can distinguish any pair of distinct actions in $T$, since all agents know 
the agreed value.   
The precondition $\precond^T(v)$ guarantees the validity condition
of the consensus task, that is, the agreed value~$v$ must be the input
to some agent~$a$.

\subsection{Solvability of the consensus task with 2 agents}
\label{subsec:2consensus}

Let us show that the consensus task is solvable, when $n=2$. %
(The case $n=1$ is trivial.)

Remember that the input simplicial model $\cplI$ for $2$ agents consists 
of four facets $X_{00}, X_{01}, X_{10}, X_{11}$, 
where $X_{ij}= \{(0,i), (1,j)\}$ for each $i,j \in \{0,1\}$.  
The consensus task is then given by 
the partial product update model 
$\ImProd{\cplI}{\actMf{T}}=\anglpair{W, \sim,  L}$ consisting of:
\begin{itemize}
    \item $W= \{(\{X_{ij}\}, 0)\mid i\neq 1 \text{ or } j\neq 1\}\cup \{(\{X\}, 1)\mid i\neq 0 \text{ or } j\neq 0\}$;
    \item The indistinguishability relation defined by  
    $(\{X\}, v) \relK{}{a} (\{X'\}, v')$ iff           $X\relK{\cplI}{a} X'$ and $v=v'$, for each $a\in\Ag$;
    \item $L\bigl((\{X\}, v)\bigr)=\labSM^{\cplI} (X)$.
\end{itemize}
The underlying complex of $\ImProd{\cplI}{\actMf{T}}$ is illustrated 
in Figure~\ref{fig:consensusPPU_2}.
Note that the complex is pure, meaning that every agent $a\in\Ag$ is alive in every world of the model.

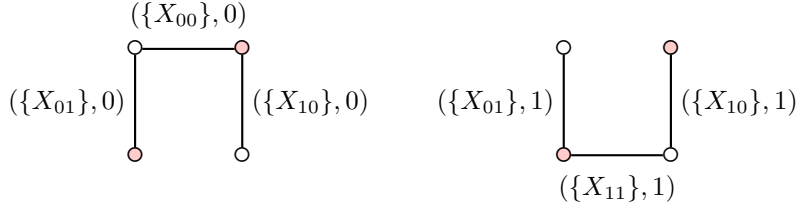
\begin{figure}[ht]
    \centering
    \begin{tikzpicture}[scale=1.41]
        \draw[thick](0, 0)--(0, 1)--(1, 1)--(1, 0);
        \draw[thick](4, 1)--(4, 0)--(5, 0)--(5, 1);

        \draw(0, 0)node{{\Large $\nodeR$}};
        \draw(0, 1)node{{\Large $\nodeW$}};
        \draw(1, 1)node{{\Large $\nodeR$}};
        \draw(1, 0)node{{\Large $\nodeW$}};

        \draw(0, 0.5)node [anchor=east]{$(\{X_{01}\}, 0)$};
        \draw(0.5, 1.1)node [anchor=south]{$(\{X_{00}\}, 0)$};
        \draw(1, 0.5)node[anchor=west]{$(\{X_{10}\}, 0)$};

        \draw(4, 1)node{{\Large $\nodeW$}};
        \draw(4, 0)node{{\Large $\nodeR$}};
        \draw(5, 0)node{{\Large $\nodeW$}};
        \draw(5, 1)node{{\Large $\nodeR$}};

        \draw(4.0, 0.5)node [anchor=east]{$(\{X_{01}\}, 1)$};
        \draw(4.5, -0. 1)node[anchor=north]{$(\{X_{11}\}, 1)$};
        \draw(5, 0.5)node[anchor=west]{$(\{X_{10}\},  1)$};
    \end{tikzpicture}
    \caption{The partial product update model $\ImProd{\cplI}{\actMf{T}}$ of the consensus task for $2$ agents}
    \label{fig:consensusPPU_2}
\end{figure}

    Next, let us consider the partial product update model 
    $\ImProd{\cplI}{\MP{\cplI}}$ for the synchronous message passing protocol, 
    where $\MP{\cplI}$ is the action model for $2$~agents. 
    Remember that the inputless action model for $2$~agents 
    is given by $\ActOneSMP=\anglpair{\MPpA{2},\MPprel{2}}$,
    where $\MPpA{2}=\{\clos{\emptyset},  \clos{\{0<1\}},$ $\clos{\{1<0\}}\}$. 
    For brevity, let us write $t_{\emptyset}$, $t_0$, and $t_1$ for the posets
    $\clos{\emptyset}$,  $\clos{\{0<1\}}$, and $\clos{\{1<0\}}$, respectively.
    Topologically, these posets correspond to the facets in the impure
    simplicial complex illustrated in Figure~\ref{fig:inputless_n2}.

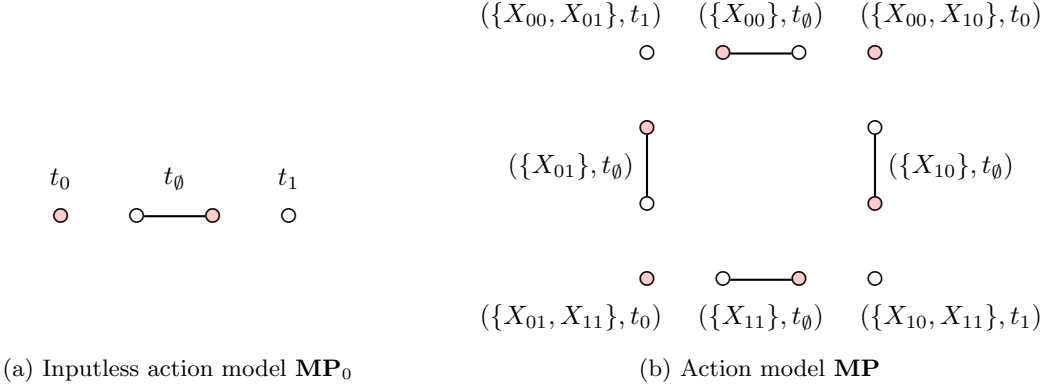
\begin{figure}[ht]
    \begin{minipage}[b]{.44\linewidth}
        \centering
        \begin{tikzpicture}
            \draw[thick](-0.5,0)--(0.5,0);

            \draw(-1.5,0)node{{\Large $\nodeR$}};
            \draw(-1.5,0.5)node{$t_0$};

            \draw(-0.5,0)node{{\Large $\nodeW$}};
            \draw(0.5,0)node{{\Large $\nodeR$}};

            \draw(1.5,0)node{{\Large $\nodeW$}};
            \draw(1.5,0.5)node{$t_1$};

            \draw(0,0.5)node{$t_{\emptyset}$};

        \end{tikzpicture}
        \vspace*{4em}
        \subcaption{Inputless action model $\ActOneSMP$}
        \label{fig:inputless_n2}
    \end{minipage}\hfil%
    \begin{minipage}[b]{.55\linewidth}
        \centering
        \begin{tikzpicture}
            \draw[thick](0, 1)--(0, 2);
            \draw[thick](1, 3)--(2, 3);
            \draw[thick](1, 0)--(2, 0);
            \draw[thick](3, 1)--(3, 2);

            \draw(0, 0)node{{\Large $\nodeR$}};
            \draw(0, 2)node{{\Large $\nodeR$}};
            \draw(1, 3)node{{\Large $\nodeR$}};
            \draw(2, 0)node{{\Large $\nodeR$}};
            \draw(3, 3)node{{\Large $\nodeR$}};
            \draw(3, 1)node{{\Large $\nodeR$}};

            \draw(0, 1)node{{\Large $\nodeW$}};
            \draw(0, 3)node{{\Large $\nodeW$}};
            \draw(1, 0)node{{\Large $\nodeW$}};
            \draw(3, 2)node{{\Large $\nodeW$}};
            \draw(3, 0)node{{\Large $\nodeW$}};
            \draw(2, 3)node{{\Large $\nodeW$}};

            \draw(1. 5, 3. 5)node{$(\{X_{00}\}, t_{\emptyset})$};
            \draw(1. 5, -0. 5)node{$(\{X_{11}\}, t_{\emptyset})$};
            \draw(-1, 3. 5)node{$(\{X_{00}, X_{01}\},  t_1)$};
            \draw(4, 3. 5)node{$(\{X_{00}, X_{10}\}, t_0)$};
            \draw(-1, 1. 5)node{$(\{X_{01}\}, t_{\emptyset})$};
            \draw(4, 1. 5)node{$(\{X_{10}\}, t_{\emptyset})$};
            \draw(-1, -0. 5)node{$(\{X_{01}, X_{11}\}, t_0)$};
            \draw(4,  -0. 5)node{$(\{X_{10}, X_{11}\},  t_1)$};
        \end{tikzpicture}
        \subcaption{Action model $\MP{\cplI}$}
        \label{fig:inputfull_n2}
    \end{minipage}
    \caption{Action model of the synchronous message passing protocol for $2$ agents}
    \label{fig:actModels_impure}
\end{figure}

Following Definition~\ref{def:SMP}, we obtain the action model 
$\MP{\cplI}=\anglpair{\MPA{\cplI}, \MPrel{\cplI}, \MPprecond{\cplI}}$ consisting of:
\begin{itemize}
    \item The set of actions 
    \[
        \MPA{\cplI}=
                \{(\{X_{ij}\}, t_{\emptyset})\mid i,j\in\{0,1\} \}
                \cup\{(\{X_{0j}, X_{1j}\},  t_0)\mid j\in \{0, 1\}\}\cup \{(\{X_{i0}, X_{i1}\},  t_1)\mid i\in \{0, 1\}\};\]
    \item The family of indistinguishability relations defined by  
    $(\ActEquiv{t}{X}, t)\MPrel{\cplI}_a (\ActEquiv{t'}{X'},t')$ iff 
    $(\ActEquiv{t}{X}, t) =(\ActEquiv{t'}{X'},t')$ and $t\sim^{\ActOneSMP}_a t'$;
    \item The labeling $\MPprecond{\cplI} \bigl((\ActEquiv{t}{X} , t)\bigr)=\bigvee \{\bigwedge \labSM^{\cplI}(X')\mid X' \in \ActEquiv{t}{X}\}$.
\end{itemize}
The corresponding topological structure of this action model is illustrated by the simplicial complex
in Figure~\ref{fig:inputfull_n2}.

    Combining $\cplI$ and $\MP{\cplI}$, we obtain a partial product update model of
    the synchronous message passing protocol $\ImProd{\cplI}{\MP{\cplI}}$ for $2$ agents,
    where its Kripke frame has a topological structure isomorphic to that of the action model $\MP{\cplI}$
    and the labeling  $L^{\ImProd{\cplI}{\MP{\cplI}}}$
    on the worlds of the Kripke frame is given by:
    \[
        L^{\ImProd{\cplI}{\MP{\cplI}}}\bigl((\ActEquiv{t}{X},  t)\bigr)= \bigcap \{\labSM^{\cplI}(X')\mid X' \in \ActEquiv{t}{X}\}.
    \]

    Observe that $\ImProd{\cplI}{\actMf{T}}$ corresponds to a complex comprising of
    two disconnected subcomponents, one for the output $0$ and the other for $1$
    and that $\ImProd{\cplI}{\MP{\cplI}}$ corresponds to
    a collection of discrete facets, none of which are connected.
    It is easy to devise a morphism $\delta:
    \ImProd{\cplI}{\MP{\cplI}}\to\cplI\{\actMf{T}\}$   that satisfies the
    condition~\eqref{eq:kripsolve} in Section~\ref{subsec:tasksolv} for
    task solvability.
    For instance, such a morphism~$\delta$ is given by:
    \begin{align*}
        \delta\bigl((\{X_{11}\},  t_{\emptyset})\bigr)  & {} = \bigsat{\{0,1\}}{(\{X_{11}\}, 1)},                           \\
        \delta\bigl((\{X_{ij}\}, t_{\emptyset})\bigr)   & {} =\bigsat{\{0,1\}}{(\{X_{ij}\}, 0)} \quad  (i\neq 1 \text{ or } j \neq 1), \\
        \delta\bigl((\ActEquiv{t_0}{X_{ij}},  t_0)\bigr) & {} =\bigsat{\{1\}}{(\{X_{ij}\},j)},                               \\
        \delta\bigl((\ActEquiv{t_1}{X_{ij}},  t_1)\bigr) & {} =\bigsat{\{0\}}{(\{X_{ij}\},  i)}.
    \end{align*}

\begin{theorem} \label{th:consSolvable2Ags}
    Suppose we have a system of $2$ agents with $\Ag=\{0,1\}$. Then
    the consensus task is solvable by
    the single-round synchronous message passing protocol.
\end{theorem}

\subsection{Unsolvability of the consensus task with 3 agents}
\label{subsec:unsolvecons3}

The consensus task is not solvable by the synchronous message passing protocol,
when the number $n$ of agents is $3$ or greater.
This is because, in contrast to the case where $n$ is $2$ or less, the
resulting partial product update model of the protocol is more tightly
connected in a topological sense.

To demonstrate this topological intuition in partial epistemic models,
we first examine the case $n=3$, where the 
partial product update model $\ImProd{\cplI}{\MP{\cplI}}$ of the 
synchronous message passing protocol is an 
impure $2$-dimensional complex. 
    Let us define a guarded positive formula $\Phi_3\in \ModLang_{\Modalfont{K}, \Palive}^+$ by
    \[
        \Phi_3 \equiv \ModK{2} \ModK{1}\ModK{2}\varphi_0 \vee \ModK{0}\ModK{2}\ModK{0} \varphi_1 \vee \ModK{1}\ModK{0}\ModK{1} \varphi_2,
    \]
    where 
    $\varphi_i \equiv \bigvee_{a\in\Ag} \Palive(a)\Rightarrow \Pinput{a}{i}$ for $i=0,1,2$.

    We claim that $\Phi_3$ is a logical obstruction.
    That is, $\Phi_3$ is valid in $\cplI\{\actMf{T}\}$ but not valid in $\ImProd{\cplI}{\ActSMP}$,
    where $\actMf{T}$ is the action model for the consensus task
    and $\ActSMP$ is the action model for the synchronous message passing protocol, for $3$~agents.

    The input simplicial complex $\cplI$ for $3$ agents has 
    an underlying pure $2$-dimensional complex and we write
    $X_{ijk}$ to denote a facet $\{(0,i), (1,j), (2,k)\}$ of the complex, where $i,j,k \in \{0, 1, 2\}$.
    The labeling on a facet is defined by $\labSM^{\cplI}(X_{ijk})=\{\Pinput{0}{i}, \Pinput{1}{j}, \Pinput{2}{k}\}$.
    The action model of the consensus task $\actMf{T}=\anglpair{T, \sim^T, \precond^T}$
    is defined similarly  as in Section~\ref{subsec:2consensus}, except that
    the set of actions $T= \{0, 1, 2\}$ consists of the three possible outputs $0$, $1$, and $2$.

    The partial product update model $\ImProd{\cplI}{\actMf{T}}$ of the consensus task
    has an underlying complex, pure of dimension~$2$.
    This means that 
    $\Palive(a)$ is true for every agent $a\in\Ag$ in every world of $\ImProd{\cplI}{\actMf{T}}$
    and therefore $\varphi_i$ 
    is equivalent to $\precond^T (i)=\bigvee_{a\in\Ag}\Pinput{a}{i}$ 
    for each action $i\in \{0,1,2\}$. 
    (Remember that $T=\Ag= \{0,\ldots, n-1\}$.)
        By the definition, $\ImProd{\cplI}{\actMf{T}}, (\preEqClass{v}{X}, v)\models \varphi_v$ holds.
    Since $(\preEqClass{v}{X}, v) \relK{T}{a} (\preEqClass{v'}{X'}, v')$ 
    implies $v=v'$ for every $a\in\Ag$,  either of the disjuncts 
    $\ModK{2} \ModK{1}\ModK{2}\varphi_0$, $\ModK{0}\ModK{2}\ModK{0} \varphi_1$, or  
    $\ModK{1}\ModK{0}\ModK{1} \varphi_2$ holds. 
    Therefore $\Phi_3$ is
    a valid formula in $\ImProd{\cplI}{\actMf{T}}$.

    We then discuss the remaining half, i.e., $\Phi_3$ is invalid in
    the partial product update model of the synchronous message passing protocol.
    Remember that the inputless action model $\ActOneSMP$ for $3$ agents 
    has an impure $2$-dimensional complex given in Figure~\ref{fig:3_protKripSmpOne}
    as the underlying complex, and that each facet in $\ActOneSMP$ is represented
    by a poset. 

    As we have seen in Section~\ref{subsec:constActModel}, 
    we can derive the action model $\MP{\cplI}$ by making 
    copies of the inputless model $\ActOneSMP$ for each input facet of the input model
    and then pasting them appropriately.
    To show the unsolvability of the consensus task, it suffices to
    consider pasting the copies for two facets $X_{012}$ and $X_{112}$.
    This  results in a subcomplex of $\MP{\cplI}$ given in
    Figure~\ref{fig:actionmodelMPI},   which merges the result of actions
    on the input   facet $X_{012}$ (the left half of the picture) and that
    on $X_{112}$     (i.e., the right half).
    In the figure, some facets are designated by the corresponding actions:
    $t_0= (\ActEquiv{\clos{\emptyset}}{X_{012}}, \clos{\emptyset})$,
    $t_1= (\ActEquiv{u_1}{X_{012}}, u_1)$, %
    $t_2= (\ActEquiv{u_1}{X_{112}}, u_1)$, %
    $t_3= (\ActEquiv{\clos{\emptyset}}{X_{112}},  \clos{\emptyset})$, and
    $t_4 = (\ActEquiv{u_2}{X_{012}}, u_2)$,
    where 
    $u_1=\clos{\{ 0< 1 \}}=\clos{\{ \nodeW < \nodeR \}}$ and
    $u_2=\clos{\{ 0<1, 0<2 \}}=\clos{\{ \nodeW< \nodeR, \nodeW < \nodeB \}}$.
    As already discussed in Section~\ref{subsec:constActModel}, 
    the two copies of the complex of $\ActOneSMP$ are merged via 
    the facet designated by $t_4$, since
    $\ActEquiv{u_2}{X_{012}}=\ActEquiv{u_2}{X_{112}}$.

    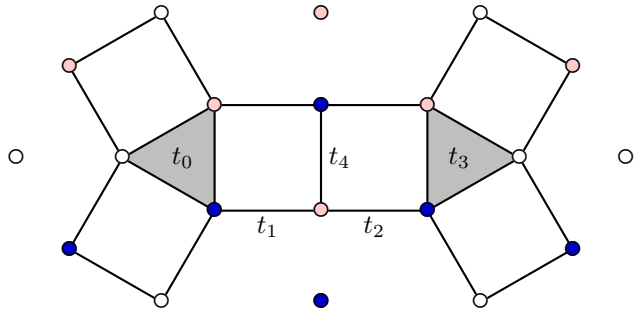
\begin{figure}[ht]
        \centering
        \begin{tikzpicture}[scale=0.7]
            \draw[fill=lightgray,  thick](2, 1)--(2, -1)--(3. 73, 0)--cycle;
            \draw(2. 6, 0)node{$t_3$};

            \draw[fill=lightgray,  thick](-2, 1)--(-2, -1)--(-3. 73, 0)--cycle;
            \draw(-2. 6, 0)node{$t_0$};

            \draw[thick](2, 1)--(-2, 1);
            \draw(-1, -1. 3)node{$t_1$};
            \draw(1, -1.3)node{$t_2$};

            \draw[thick](2, -1)--(-2, -1);
            \draw[thick](0, 1)--(0, -1);
            \draw(0.35, 0)node{$t_4$};

            \draw[thick](2, 1)--(3, 2. 73)--(4. 73, 1. 73)--(3. 73, 0);
            \draw[thick](2, -1)--(3, -2. 73)--(4. 73, -1. 73)--(3. 73, 0);

            \draw[thick](-2, 1)--(-3, 2. 73)--(-4. 73, 1. 73)--(-3. 73, 0);
            \draw[thick](-2, -1)--(-3, -2. 73)--(-4. 73, -1. 73)--(-3. 73, 0);

            \draw(0, 1)node{{\Large $\nodeB$}};
            \draw(0, -1)node{{\Large $\nodeR$}};

            \draw(0, 2.73)node{{\Large $\nodeR$}};
            \draw(0, -2.73)node{{\Large $\nodeB$}};
            \draw(-5.73, 0)node{{\Large $\nodeW$}};
            \draw(5.73, 0)node{{\Large $\nodeW$}};

            \draw(2, 1)node{{\Large $\nodeR$}};
            \draw(2, -1)node{{\Large $\nodeB$}};
            \draw(3. 73, 0)node{{\Large $\nodeW$}};
            \draw(3, 2. 73)node{{\Large $\nodeW$}};
            \draw(4. 73, 1. 73)node{{\Large $\nodeR$}};
            \draw(3, -2. 73)node{{\Large $\nodeW$}};
            \draw(4. 73, -1. 73)node{{\Large $\nodeB$}};

            \draw(-2, 1)node{{\Large $\nodeR$}};
            \draw(-2, -1)node{{\Large $\nodeB$}};
            \draw(-3. 73, 0)node{{\Large $\nodeW$}};
            \draw(-3, 2. 73)node{{\Large $\nodeW$}};
            \draw(-4. 73, 1. 73)node{{\Large $\nodeR$}};
            \draw(-3, -2. 73)node{{\Large $\nodeW$}};
            \draw(-4. 73, -1. 73)node{{\Large $\nodeB$}};
        \end{tikzpicture}
        \caption{A subcomplex of the action model $\MP{\cplI}$ that merges the actions on $X_{012}$ and $X_{112}$}
        \label{fig:actionmodelMPI}
    \end{figure}

    We obtain the partial product update model $\ImProd{\cplI}{\MP{\cplI}}$ 
    for the synchronous message passing protocol by combining $\cplI$ with $\MP{\cplI}$.
    We note that $\ImProd{\cplI}{\MP{\cplI}}$
    has an isomorphic partial epistemic frame with that of the action model $\MP{\cplI}$.
    (See Appendix~\ref{prop:protocolEq} for the formal discussion of this fact.)
    Therefore, without loss of generality,
    we may regard $\ImProd{\cplI}{\MP{\cplI}}$
    as a partial epistemic model $\anglpair{\MPA{\cplI}, \MPrel{\cplI},L^{\MP{\cplI}}}$,
    where $\anglpair{\MPA{\cplI}, \MPrel{\cplI}}$ is the Kripke frame of $\MP{\cplI}$
    and $L^{\MP{\cplI}}$ is a labeling on the worlds of $\MPA{\cplI}$:
    \[
        L^{\MP{\cplI}}\bigl((\ActEquiv{t}{X},  t)\bigr)= \bigcap \{\labSM^{\cplI}(X')\mid X' \in \ActEquiv{t}{X}\}.
    \]

    Now we show that $\Phi_3$ is invalid in 
    the partial product update model $\ImProd{\cplI}{\MP{\cplI}}$.
    We claim that $\Phi_3$ is false in the world $t_0$ in Figure~\ref{fig:actionmodelMPI}.

    From Figure~\ref{fig:actionmodelMPI}, it is easy to see that
    the facet $t_0$ is connected with $t_3$ via
    $t_0 \MPrel{\cplI}_2 t_1 \MPrel{\cplI}_1 t_2 \MPrel{\cplI}_2 t_3$.
    We show $\ImProd{\cplI}{\MP{\cplI}}, t_0 \not\models \ModK{2}\ModK{1}\ModK{2} \varphi_0$, 
    by traversing the above path in reverse order.
    In $t_3$, all agents are alive, that is,
    $\ImProd{\cplI}{\MP{\cplI}},  t_3 \models \Palive(a) $
    for every $a\in \Ag$.
    Since $L^{\ActSMP}(t_3)=  \{\Pinput{0}{1}, \Pinput{1}{1}, \Pinput{2}{2}\}$,
    $\varphi_0 \equiv \bigvee_{a\in\Ag} \Palive(a)\Rightarrow \Pinput{a}{0}$ is false,
    i.e., $\ImProd{\cplI}{\MP{\cplI}},  t_3 \not\models \varphi_0$.
    Since $t_3 \MPrel{\cplI}_2 t_2$,
    this implies that $\ImProd{\cplI}{\MP{\cplI}},  t_2 \not\models \ModK{2}\varphi_0$.
    Repeating this by traversing $t_2 \MPrel{\cplI}_1 t_1$ and then
    $t_1 \MPrel{\cplI}_2 t_2$,
    we obtain
    $\ImProd{\cplI}{\MP{\cplI}},  t_2 \not\models \ModK{2}\ModK{1}\ModK{2}\varphi_0$.
    This shows that $\ModK{2}\ModK{1}\ModK{2}\varphi_0$ is an invalid formula in
    $\ImProd{\cplI}{\MP{\cplI}}$.

    Due to the symmetry of the partial product update model $\ImProd{\cplI}{\MP{\cplI}}$,
    we can similarly show the invalidity of the remaining disjuncts in $\Phi_3$.
    Therefore $\Phi_3$ is invalid in $\ImProd{\cplI}{\MP{\cplI}}$.

\begin{theorem}\label{thm:unsolvecons3}
    For a system of $3$ agents with $\Ag=\{0,1,2\}$, the consensus task is not solvable by
    the single-round synchronous message passing protocol.
\end{theorem}

\subsection{Unsolvability of the consensus task, general case}
\label{subsec:unsolveconsgeneral}

We show the unsolvability for the general case, where the number of agents is $3$ or greater.
The proof is almost the same as that for $3$~agents, but we have to argue
the connectivity in a complex of higher dimension. 
Below we present a logical obstruction for the general case and
discuss the connectivity symbolically on the formal description of actions.

\begin{theorem}\label{thm:unsolveconsn}
    For a system with $\Ag=\{0,\cdots, n-1\}$ ($n\geq3$), 
    the consensus task is not solvable by the single-round synchronous message passing protocol.
\end{theorem}

\begin{proof}
    Let us define a guarded positive formula $\Phi_n$ ($n\geq 3$) by:
    \[
        \Phi_n \equiv \bigvee_{i\in \{0, \cdots,  n-1\}} \ModK{(i+2)\bmod n}\ModK{(i+1)\bmod n}\ModK{(i+2)\bmod n} \varphi_i,
    \]
    where $\varphi_i \equiv \bigvee_{a\in\Ag} \Palive(a)\Rightarrow \Pinput{a}{i}$
    for each $i$ ($0\leq i\leq n-1$).

    We claim that $\Phi_n$ is a logical obstruction.
    Let $\cplI$ be the input model, $\actMf{T}$ be the action model of the consensus task,
    $\MP{\cplI}$ be the action model of the synchronous message passing protocol, 
    each defined for $n$ agents.
    We also let $\ImProd{\cplI}{\actMf{T}}$
    and $\ImProd{\cplI}{\MP{\cplI}}$ be the partial product update models
    for the task and the protocol, respectively.

    By Theorem~\ref{thm:LOTheorem}, it suffices to
    show that $\Phi_n$ is valid in $\ImProd{\cplI}{\actMf{T}}$
    but is invalid in $\ImProd{\cplI}{\MP{\cplI}}$.
    The validity of $\Phi_n$ is $\ImProd{\cplI}{\actMf{T}}$ is similarly shown
    as we have demonstrated in Section~\ref{subsec:unsolvecons3}.

    Let us show that $\Phi_n$ is invalid in 
    $\ImProd{\cplI}{\MP{\cplI}}=\anglpair{\MPA{\cplI}, \MPrel{\cplI},L^{\ActSMP}}$. 
    Similarly as in Section~\ref{subsec:unsolvecons3}, by symmetry,
    it suffices to show that $\ModK{2} \ModK{1} \ModK{2} \varphi_0$
    is false in some world of $\ImProd{\cplI}{\MP{\cplI}}$.
    We show that $\ImProd{\cplI}{\MP{\cplI}}, (\ActEquiv{\clos{\emptyset}}{X},  \clos{\emptyset})
        \not\models \ModK{2} \ModK{1} \ModK{2} \varphi_0$,
    where $X$ is a facet of $\cplI$ that is uniquely identified by the labeling
    $\labSM^{\cplI} (X)=\{\Pinput{i}{i}\mid 0\leq i \leq n-1\}$.

    Let us define $X'$ be the facet of $\cplI$ that is uniquely determined by
    the labeling $\labSM (X') =\{\Pinput{0}{1}\} \cup \{\Pinput{i}{i}\mid 1\leq i \leq n-1\}$.
    We show that
    $(\ActEquiv{\clos{\emptyset}}{X},  \clos{\emptyset})$
    and
    $(\ActEquiv{\clos{\emptyset}}{X'},  \clos{\emptyset})$
    are connected along the following path:
    \[
        (\ActEquiv{\clos{\emptyset}}{X}, \clos{\emptyset})\MPrel{\cplI}_2
        (\ActEquiv{\clos{\{0< 1\}}}{X},  \clos{\{0<1\}})\MPrel{\cplI}_1
        (\ActEquiv{\clos{\{0< 1\}}}{X'}, \clos{\{0<1\}}) \MPrel{\cplI}_2
        (\ActEquiv{\clos{\emptyset}}{X'},  \clos{\emptyset}).
    \]

    To show $(\ActEquiv{\clos{\emptyset}}{X}, \clos{\emptyset})\MPrel{\cplI}_2
        (\ActEquiv{\clos{\{0< 1\}}}{X},  \clos{\{0<1\}})$, 
    it suffices to show $\clos{\emptyset} \relK{\ActOneSMP}{2} \clos{\{0<1\}}$, 
    which immediately follows from $2\not\in \lo(\clos{\emptyset})\cup \lo(\clos{\{0<1\}})$.
    The proof of 
    $(\ActEquiv{\clos{\{0< 1\}}}{X'}, \clos{\{0<1\}}) \MPrel{\cplI}_2
    (\ActEquiv{\clos{\emptyset}}{X'},  \clos{\emptyset})$ is similar.
    For $(\ActEquiv{\clos{\{0< 1\}}}{X},  \clos{\{0<1\}})\MPrel{\cplI}_1 (\ActEquiv{\clos{\{0< 1\}}}{X'}, \clos{\{0<1\}})$, 
    since $X\not\sim^{\cplI}_a X'$ if and only if $a\neq 0$,  
    it suffices to show that $\clos{\{0<1\}}\relK{\ActOneSMP}{1} \clos{\{0<1\}}$, 
    which  follows from $1 \not\in \Lo(\clos{\{0<1\}})$.

    All agents are alive in $(\ActEquiv{\clos{\emptyset}}{X'},  \clos{\emptyset})$, that is,
    $\MP{\cplI},  (\ActEquiv{\clos{\emptyset}}{X'},  \clos{\emptyset}) \models \Palive(a)$ 
    holds   for every $a\in \Ag$.
    Furthermore, 
    $L^{\ImProd{\cplI}{\MP{\cplI}}} \bigl((\ActEquiv{\clos{\emptyset}}{X'},  \clos{\emptyset})\bigr)
    = \{ \Pinput{0}{1} \} \cup \{ \Pinput{i}{i} \mid 1\leq i \leq n-1  \}$.
    These imply $\ImProd{\cplI}{\MP{\cplI}}, (\ActEquiv{\clos{\emptyset}}{X'},  \clos{\emptyset}) \not\models
    \varphi_0$. %
    Traversing the above path in reverse order, we sequentially obtain 
    $\ImProd{\cplI}{\MP{\cplI}},(\ActEquiv{\clos{\{0< 1\}}}{X'}, \clos{\{0<1\}})$ $\not\models \ModK{2}\varphi_0$, then
    $\ImProd{\cplI}{\MP{\cplI}}, (\ActEquiv{\clos{\{0< 1\}}}{X}, \clos{\{0<1\}}) \not\models \ModK{1}\ModK{2}\varphi_0$, 
    and finally
    $\ImProd{\cplI}{\MP{\cplI}}, (\ActEquiv{\clos{\emptyset}}{X},  \clos{\emptyset}) \not\models \ModK{2}\ModK{1}\ModK{2}\varphi_0$. 
    This proves that $\ModK{2}\ModK{1}\ModK{2}\varphi_0$ is invalid in  $\ImProd{\cplI}{\MP{\cplI}}$.

    Consequently, by Theorem~\ref{thm:LOTheorem},
    we conclude that the consensus task is not solvable by the synchronous message passing protocol,
    if the number of agents is $3$ or greater.
\end{proof}

\section{Conclusion and Future Work}
\label{sec:conclusion}

We have proposed the notion of partial product update models that is suitable 
for modeling distributed tasks and protocols in partial epistemic models,
in which agents may die. 
Using partial product update models, 
we defined task solvability in partial epistemic models,
thereby providing the logical method that allows a logical 
obstruction to prove the unsolvability of a distributed task.  
This logical method defined for partial product update models
extends the original product update proposed for epistemic models:
Given a pair of an input model $\cplI$ and an action model $\actMf{A}$, 
a partial product update model $\ImProd{\cplI}{\actMf{A}}$ refines 
the original product update model up to indistinguishability by the set of live agents.
The unsolvability of a task is then proved 
by a logical obstruction.
We have presented a concrete formula of epistemic logic to show that the consensus task is
unsolvable by the single-round synchronous message passing protocol. 
We have shown that the formula is indeed a logical obstruction, 
where the partial product update model for the synchronous message passing protocol 
is constructed from an action model whose actions are posets of rank at most $1$.

We have demonstrated an unsolvability result, but only for the consensus task and the single-round execution of the protocol. 
In contrast, the topological method can prove more general unsolvability results 
for $k$-set agreement tasks and multiple-round execution of the protocol \cite{S1571:2001:HerlihyRajsbaumTuttle}. 
It is an interesting topic to pursue how to obtain these generalized results 
using the logical method. 

The notion of partial product update itself could be of interest from the perspective of 
dynamic epistemic logic \cite{DitmarschHoekKooi:DELbook08}. 
The original product update was developed 
in the context of dynamic updates of Kripke models: 
In terms of logic with epistemic action modalities, 
$\cplI, X \models [\actMf{A}] \varphi$ iff
$\cplI, X \models \precond(t)$ implies
$\Prod{\cplI}{\actMf{A}}, (X,t) \models \varphi$. 
It would be interesting 
to figure out an appropriate logical interpretation of partial product update.

\subsection*{Acknowledgment}
The third author is supported by JSPS KAKENHI Grant Number 20K11678.

{
\bibliographystyle{plainurl}
\bibliography{distrib}
}

\appendix 
\section*{Appendix}
\section{Equivalence of The Two Task Solvabilities}
\label{sec:solvablEq}

\subsection{The topological definition of task solvability}
\label{subsec:toptasksolve}

In the topological method for distributed computing, tasks and protocols are defined
by means of carrier maps, where a carrier map is a color-preserving function that associates each 
\emph{simplex} of an input complex with a subcomplex of an output complex \cite{Book:2013:HerlihyKozlovRajsbaum}.
However, the partial epistemic models in this paper 
are built around facets, not simplexes of arbitrary dimension, 
which implies that partial product update models cannot cover 
all aspects of carrier maps. 
(Partial product update models assume that all agents participate in the computation anyway, 
though some of them may crash during the execution.)

For this reason, in this paper we define the topological specification of
a task or a protocol by a suboptimal function, which we call a
\keywd{facet map}, that associates each facet of the input complex with a
set of facets of the output complex.

In what follows, we assume the set of vertexes $V$ of a complex $\cplC=\anglpair{V,S,\coloring}$
is a finite subset of $\Ag\times\Value$ and the coloring map 
is defined by $\coloring\bigl((a,v)\bigr)=a$ for each $(a,v)\in V$. 
With this assumption, we can easily translate
a complex $\anglpair{V,S,\coloring}$ into the corresponding
input simplicial model $\anglpair{V,S,\coloring,\labSM}$ given in Definition~\ref{def:inputsimplicialmodel}, 
by defining $\labSM(X)=\{\Pinput{a}{v} \mid (a,v)\in X\}$,
and similarly for the reverse translation.
In this way, 
we will confuse a complex with an input simplicial model.

\begin{definition}
    \label{def:facetmap}
    Given complexes 
    $\cplC=\anglpair{V,S,\coloring}$ and $\cplD=\anglpair{V',S',\coloring'}$, 
    a function $\Theta: \Facet(\cplC) \to 2^{\Facet(\cplD)}$ is  
    called a \keywd{facet map} if it satisfies the following conditions.
    \begin{itemize}
        \item $\Theta$ is color-preserving, that is, 
        $\displaystyle\bigcup_{Y\in \Theta(X)}\coloring'(Y)\subseteq \coloring(X)$ for every facet $X\in \Facet(\cplC)$;
        \item $\Theta$ is surjective, i.e., 
         $\Facet(\cplD) = \bigcup_{X\in\Facet(\cplC)} \Psi(X)$.
    \end{itemize}

    In order to define a task and a protocol, 
    let $\cplI=\anglpair{V^{\cplI}, S^{\cplI}, \coloring^{\cplI}}$ be the common input complex of them.
    \label{def:simplicialtask}
    A task is defined by a triple $\anglpair{\cplI,\cplO,\Delta}$, which we call a \keywd{simplicial task},
    where 
        $\cplO=\anglpair{V^{\cplO}, S^{\cplO}, \coloring^{\cplO}}$ is an output complex and 
        $\Delta:\func{\Facet(\cplI)}{\PowerSet{\Facet(\cplO)}}$ is a facet map.
    \label{def:simplicialprotocol}
    A protocol is defined by a triple $\anglpair{\cplI,\cplP,\Psi}$, which we call a \keywd{simplicial protocol},
    where 
    $\cplP=\anglpair{V^{\cplP}, S^{\cplP}, \coloring^{\cplP}}$ is an output complex and 
    $\Psi:\func{\Facet(\cplI)}{\PowerSet{\Facet(\cplP)}}$ is a facet map satisfying the following condition:

        \begin{align} \label{eq:carriercont} & 
            \begin{minipage}{.85\textwidth}
                $\coloring^{\cplP}(Y\cap Y')\subseteq \coloring^{\cplI} (X\cap X')$,
                for any $X,X'\in \Facet(\cplI)$ and $Y, Y' \in \Facet(\cplP)$  such that  \\
                $Y\in \Psi(X)$ and $Y'\in \Psi(X')$. 
            \end{minipage}
        \end{align}
\end{definition}

The condition~\eqref{eq:carriercont} is a natural requirement for the protocol:
If an agent~$a$ observes the same output of the protocol in two facets $Y$ and $Y'$, 
the agent~$a$ must also have agreed on the input in two facets $X$ and $X'$, 
where  $Y$ (resp., $Y'$) represents a possible global state of the system 
that is reachable from the global state of the initial input represented by $X$ (resp., $X'$). 

\begin{definition}\label{def:topsolvability}
    We say 
    a protocol $\anglpair{\cplI, \cplP, \Psi}$ solves a task $\anglpair{\cplI, \cplO, \Delta}$,  
    if there exists a \keywd{descision map}, i.e., 
    a chromatic simplicial map $\DeltaTop: \cplP \to \cplO$ that satisfies
    \begin{align} \label{eq:topsolve} &
    \begin{minipage}{.85\textwidth}
            $\DeltaTop (\Psi(X)) \subseteq \Delta(X)$ for every $X \in \Facet(\cplI)$,  
    \end{minipage}          
    \end{align}
    where  
    $\DeltaTop (\Psi(X)) = \{ \DeltaTop(Y) \mid Y \in \Psi(X) \}$, and
    the inclusion $\mathcal{S} \subseteq \mathcal{S'}$
    means that for any $X\in \mathcal{S}$
    there exists $X'\in\mathcal{S'}$ such that $X\subseteq X'$. 
\end{definition}

\subsection{Equivalence of two task solvabilities}\label{subsec:solvablEq}

We will show that the above topological task solvability using simplicial complexes 
can be translated to that using partial product update models 
(Definition~\ref{subsec:tasksolv}).
For this, we introduce a translation function $\actionkappa$ as below.

\begin{definition}\label{def:actionkappa}
    Let $\cplI=\anglpair{V^{\cplI}, S^{\cplI}, \coloring^{\cplI}, \labSM^{\cplI}}$ be an input simplicial model.   
    We define a translation function $\actionkappa$ that assigns an action model for 
    a simplicial task or protocol $\anglpair{\cplI, \cplD, \Theta}$
    by 
    \[ \actionkappa(\anglpair{\cplI, \cplD, \Theta})= \actMf{A} \]
    where $\actMf{A}=\anglpair{A, \sim^A, \precond^A}$ is an action model consisting of:
    \begin{itemize}
        \item The set of actions $A = \Facet(\cplD)$;
        \item The family of indistinguishability relations defined by   
        $X\sim_a^A Y$ iff $a\in \coloring^{A} (X\cap Y)$;
        \item The precondition defined by 
        $\precond^A (Y)=\bigvee\{\bigwedge \labSM^{\cplI} (X)\mid X\in \Facet(\cplI),  Y\in \Theta(X)\}$.
    \end{itemize}
\end{definition}

The translation function $\actionkappa$ is an augmentation of 
the functor $\kappa$, which has been introduced 
in \cite{STACS22:GoubaultLedentRajsbaum} to show the association of 
simplicial complexes with their corresponding partial epistemic frames,
for the purpose of establishing the categorical equivalence of these two structures.
As such, $\actionkappa$ preserves the partial epistemic frame $\anglpair{\Facet(\cplD),\sim^{\cplD}}$
induced from $\cplD$. Furthermore, the precondition 
$\precond^A$ gives the condition for $Y\in \Facet(\cplD)$ to be
a possible output of $\Theta$ for an input $X \in \Facet(\cplI)$, that is,  
$\cplI,  X\models \precond^A (Y)$ if and only if $Y\in \Theta(X)$. 
In this sense, $\actionkappa$ associates an action model that is equivalent to a given 
simplicial task or protocol.

Particularly for an action model $\actMf{P}$ that is translated from a simplicial protocol $\anglpair{\cplI,\cplP,\Phi}$, 
the partial product update model $\ImProd{\cplI}{\actMf{P}}$ has a partial epistemic frame 
that is isomorphic to that of $\actMf{P}$.
\begin{proposition}\label{prop:protocolEq}
    Let $\cplI=\anglpair{V^{\cplI}, S^{\cplI}, \coloring^{\cplI}, \labSM^{\cplI}}$ be an input simplicial model,
    $\anglpair{\cplI, \cplP, \Psi}$ be 
    a simplicial protocol,  and $\actMf{P}=\actionkappa(\anglpair{\cplI, \smplMf{P}, \Psi})$
    be the translated action model. Then 
    the action model $\actMf{P}$ and the partial product update model $\ImProd{\cplI}{\actMf{P}}$ 
    have isomorphic partial epistemic frames.
\end{proposition}

\begin{proof}
Let $\actMf{P}=\anglpair{P, \sim^P, \precond^P}$
and $\ImProd{\cplI}{\ActSMP} = \anglpair{W^{\ImProd{\cplI}{\ActSMP}}, \relK{\ImProd{\cplI}{\ActSMP}}{}, L^{\ImProd{\cplI}{\ActSMP}}}$. 
We define a pair of maps 
$f: P \to W^{\ImProd{\cplI}{\ActSMP}}$ 
and 
$g: W^{\ImProd{\cplI}{\ActSMP}} \to P$ 
by 
\[
    f(Y)=(\preEqClass{Y}{X},  Y) ~~ \text{and} ~~ g\bigl((\preEqClass{Y}{X}, Y)\bigr)= Y,  
    \]
where $X,Y \in \Facet(\cplI)$. 

We show that $f$ is well-defined, that is, 
the set $\preEqClass{Y}{X}$ is not affected by the choice of $X$. 
Suppose $(\preEqClass{Y}{X},  Y), (\preEqClass{Y}{X'},  Y) \in W^{\ImProd{\cplI}{\ActSMP}}$. 
By the definition of $\preEqClass{Y}{-}$, we have $\cplI,X\models \precond^P(Y)$ and
$\cplI,X' \models \precond^P(Y)$, which implies $Y\in \Psi(X)$ and $Y\in \Psi(X')$. 
Then, by the condition~\eqref{eq:carriercont} of the simplicial protocol, 
$X\relK{\cplI}{a}X'$ for all $a\in \AliveSet{Y}$ and
therefore $\preEqClass{Y}{X}=\preEqClass{Y}{X'}$. 

It is easy to see that $f$ and $g$ are inverses of each other, namely,  
$f\circ g=\ident{W^{\ImProd{\cplI}{\ActSMP}}}$ and $g\circ f =\ident{P}$.
Furthermore they preserve the indistinguishability relation. 
For $f$, suppose $Y\relK{P}{a} Y'$, where 
$f(Y)=(\preEqClass{Y}{X},Y)$ and $f(Y')=(\preEqClass{Y'}{X'},Y')$. 
By a similar discussion above, it follows that $Y\in \Psi(X)$ and $Y'\in \Psi(X')$
from the definition of $\preEqClass{Y}{-}$. 
This implies that $X \relK{\cplI}{a} X'$ and therefore 
$f(Y)\relK{\ImProd{\cplI}{\actMf{P}}}{a} f'(Y')$.
For the converse, 
suppose $\preEqClass{Y}{X}\relK{\ImProd{\cplI}{\actMf{P}}}{a} \preEqClass{Y'}{X'}$. 
Then, $Y \relK{\cplI}{a} Y'$ and therefore $g(Y) \relK{\cplI}{a} g(Y')$. 

The isomorphisms $f$ and $g$ can be regarded as morphisms over partial epistemic frames, 
defining 
$\hat{f}(Y)=\sat{\AliveSet{Y}}{f(Y)}$ and 
$\hat{g}\bigl((\preEqClass{Y}{X}, Y)\bigr)=\bigsat{\AliveSet{Y}}{g\bigl((\preEqClass{Y}{X}, Y)\bigr)}$. 
Hence $\anglpair{P,\sim^P}$ and $\ImProd{\cplI}{\actMf{P}}$ are isomorphic partial epistemic frames.
\end{proof}

We show the task solvability defined using simplicial complexes
implies that defined using partial product update models. 
That is, the existence of a morphism $\dectop$ %
in Definition~\ref{def:kripkesolvability} 
implies the existence of a morphism $\deckrip$ %
in Definition~\ref{def:topsolvability}, and vice versa.

\begin{lemma}\label{lem:kriptosimpldecision}
    Suppose a simplicial task $\anglpair{\cplI, \cplP, \Psi}$ 
    and a simplicial protocol $\anglpair{\cplI, \cplO, \Delta}$ are given, where
    $\cplI=\anglpair{V^{\cplI}, S^{\cplI}, \coloring^{\cplI}, \labSM^{\cplI}}$ is an input simplicial model.
    Let $\actMf{P}=\actionkappa(\anglpair{\cplI, \cplP, \Psi})$ and 
    $\actMf{T}=\actionkappa(\anglpair{\cplI, \cplO, \Delta})$ be the action models for the protocol and the task, respectively. 
    If there exists a morphism over partial epistemic models $\deckrip:\func{\ImProd{\cplI}{\actMf{P}}}{\ImProd{\cplI}{\actMf{T}}}$ 
    that satisfies the condition~\eqref{eq:kripsolve} in Definition~\ref{def:kripkesolvability}, 
    there exists a simplicial map $\dectop:\func{\cplP}{\cplO}$ that satisfies the condition~\eqref{eq:topsolve}
    in Definition~\ref{def:topsolvability}.
\end{lemma}
\begin{proof}
Let $\actMf{P}=\anglpair{P, \sim^P, \precond^P}$ and $\actMf{T}=\anglpair{T, \sim^T, \precond^T}$
and also let $\cplP = \anglpair{V^{\cplP}, S^{\cplP}, \coloring^{\cplP}}$
and $\cplO = \anglpair{V^{\cplO}, S^{\cplO}, \coloring^{\cplO}}$.

Using $\deckrip$, 
we define a map $\dectop:V^{\cplP} \to V^{\cplO}$ by
\[
    \dectop(v) = u, 
    \]
where $v$ and $u$ is any pair of vertexes such that
$\coloring^{\cplP}(v) = \coloring^{\cplO}(u)$
and furthermore 
$v\in X$, $u\in Y$, and 
$(\preEqClass{Y}{Z'}, Y) \in \deckrip\bigl((\preEqClass{X}{Z}, X)\bigr)$ holds 
for some $X\in\Facet(\actMf{P})$, $Y\in\Facet(\actMf{O})$, and $Z, Z'\in \Facet(\cplI)$. 

We show $\coloring^{\cplP}(v)$ is defined for every $v\in V^{\cplP}$. 
Suppose $X$ is a facet satisfying $v\in X$. 
By the surjectivity of $\Psi$, there exists $Z\in\Facet(\cplI)$ such that 
$X\in \Psi(Z)$. This implies $(\preEqClass{X}{Z}, X)$ is defined.
Furthermore we have $\AliveSet{X}\subseteq\AliveSet{Y}$,   
since $\deckrip$ preserves the indistinguishability, $\AliveSet{X}\subseteq\AliveSet{Y}$. 
Let us define $\coloring^{\cplP}(v) =u$, 
where  $u\in Y$ be the vertex such that $\coloring^{\cplO}(u)=\coloring^{\cplP}(v)$. 
We show this uniquely defines $u$.
Suppose $u_1, u_2 \in V^{\cplO}$ are vertexes given by the definition, that is,
$\coloring^{\cplP}(v) = \coloring^{\cplO}(u_1)= \coloring^{\cplO}(u_2)$ and
furthermore there exists $X_1, X_2\in\Facet(\cplP)$, 
$Y_1, Y_2\in\Facet(\cplO)$, and
$Z_1, Z_2, Z_1', Z_2' \in\Facet(\cplI)$ such that 
$v\in X_i$, $u_i\in Y_i$, and 
$(\preEqClass{Y_i}{Z_i'}, Y_i) \in \deckrip\bigl((\preEqClass{X_i}{Z_i}, X_i)\bigr)$
for each $i=1,2$.  
By the definition of~$\deckrip$, we 
have $\AliveSet{X_i}\subseteq\AliveSet{Y_i}$ ($i=1,2$) and hence 
$\AliveSet{X_1\cap X_2}\subseteq \AliveSet{Y_1\cap Y_2}$. 
Since $v\in X_1 \cap X_2$,  
$\coloring^{\cplP}(v) = \coloring^{\cplO}(u_1)= \coloring^{\cplO}(u_2) \in \AliveSet{Y_1\cap Y_2}$.
This implies $u_1, u_2\in Y_1 \cap Y_2$ and hence $u_1=u_2$. 

To show that $\dectop$ defined above is a simplicial map, suppose $X = \{v_0,\ldots,v_m\}\in \Facet(\cplP)$. 
By the surjectivity of $\Psi$, 
there exists $Z\in \Facet(\cplI)$ such that $X\in\Psi(Z)$. 
This implies $(\preEqClass{X}{Z}, X) \in \Facet(\ImProd{\cplI}{\actMf{P}})$. 
By the condition~\eqref{eq:kripsolve} for $\deckrip$, 
$(\preEqClass{Y}{Z'}, Y)\in \deckrip\bigl((\preEqClass{X}{Z}, X)\bigr)$
for some $Y\in \Facet(\cplO)$ and $Z'\in \Facet(\cplI)$. 
By the definition of~$\dectop$, for each $v_i$, we have
$\dectop(v_i)=u_i$ where $u_i$ is the unique vertex satisfying $u_i\in Y$ 
and $\coloring^{\cplO}(u_i)= \coloring^{\cplP}(v_i)$. 
Therefore $\dectop(X)\subseteq Y\in \Facet(\cplO)$.

Let us show that $\dectop$ also satisfies \eqref{eq:topsolve}. 
Suppose $X\in\Facet(\cplI)$ and also $Y\in\Psi(X)$.
The latter implies $\cplI,X\models \precond^{\actMf{P}}(Y)$
and therefore 
$(\preEqClass{Y}{X}, Y)$
is defined. 
Then by~\eqref{eq:kripsolve}, there exists    
$(\preEqClass{Y'}{X'}, Y')\in \deckrip\bigl((\preEqClass{Y}{X}, Y)\bigr)$
such that $\preEqClass{Y}{X}\subseteq \preEqClass{Y'}{X'}$. 
Since $X\in  \preEqClass{Y}{X}\subseteq \preEqClass{Y'}{X'}$, 
$\preEqClass{Y'}{X'}=\preEqClass{Y'}{X'}$ 
and henceforth $(\preEqClass{Y'}{X'}, Y')\in \Facet(\ImProd{\cplI}{\actMf{T}})$.
This implies that $\cplI, X\models\precond^{\actMf{T}}(Y')$ and hence 
$Y'\in \Delta(X)$. 
By the definition of~$\dectop$, $\dectop(X)\subseteq Y'$. 
Therefore $\dectop(\Psi(X))\subseteq \Delta(X)$. 
\end{proof}

The converse direction is proven as follows. 

\begin{lemma}\label{lem:smpltokripdecision}
    Let $\actMf{P}=\actionkappa(\anglpair{\cplI, \cplP, \Psi})$ and  
    $\actMf{T}=\actionkappa(\anglpair{\cplI, \cplO, \Delta})$ be 
    the translated action model for a simplicial protocol $\anglpair{\cplI, \cplP, \Psi}$ 
    and a simplicial task $\anglpair{\cplI, \cplO, \Delta}$, respectively. 
    If there exists a simplicial map $\dectop:\func{\cplP}{\cplO}$ that satisfies \eqref{eq:topsolve}, 
    then there exists a morphism over partial epistemic models $\deckrip:\func{\ImProd{\cplI}{\actMf{P}}}{\ImProd{\cplI}{\actMf{T}}}$ that satisfies \eqref{eq:kripsolve}.
\end{lemma}
\begin{proof}
Let $\actMf{P}=\anglpair{P, \sim^P, \precond^P}$ and $\actMf{T}=\anglpair{T, \sim^T, \precond^T}$.
Using $\dectop$, we define $\deckrip:\func{\ImProd{\cplI}{\actMf{P}}}{\ImProd{\cplI}{\actMf{T}}}$ by
\[ 
    \deckrip \bigl((\preEqClass{Y}{X}, Y)\bigr)= \bigsat{\AliveSet{Y}}{(\preEqClass{Z}{X}, Z)}, 
    \]
where $Z\in T$, $\dectop(Y)\subseteq Z$, $\AliveSet{Z}\subseteq\AliveSet{X}$, and 
$\cplI, X\models \precond^T (Z)$. 

There exists $Z\in T$ that satisfies this condition. 
By the definition of $\preEqClass{Y}{X}$, we have 
$\cplI, X\models \precond^P (Y)$, which implies $Y\in \Psi(X)$.
By the condition~\eqref{def:topsolvability}, there exists $Z\in\Delta(X)$ 
such that $\dectop(Y)\subseteq Z$. This means $Z\in T$ and $\cplI, X\models \precond^T(Z)$. 
Furthermore, $\AliveSet{Z}\subseteq\AliveSet{X}$ because $\Delta$ is color-preserving.

The definition of $\deckrip$ is not affected by the choice of $Z$. 
To show this, suppose $Z_1, Z_2 \in T$ such that 
$\dectop(Y)\subseteq Z_i$, $\AliveSet{Z_i}\subseteq\AliveSet{X}$, and $\cplI, X\models \precond^T (Z_i)$ 
($i=1,2$). 
From $\dectop(Y)\subseteq Z_i$ ($i=1,2$), we have
$\coloring^T(Z_1\cap Z_2) \supseteq  \coloring^T(\dectop(Y)) =\coloring^P(Y)$. 
Hence $Z_1 \sim_{\AliveSet{Y}}^T Z_2$. 
This implies $\preEqClass{Z_1}{X} \sim_{\AliveSet{Y}}^{\ImProd{\cplI}{\actMf{T}}} \preEqClass{Z_2}{X}$
and therefore  $\bigsat{\AliveSet{Y}}{(\preEqClass{Z_1}{X}, Z_1)} =\bigsat{\AliveSet{Y}}{(\preEqClass{Z_2}{X}, Z_2)}$.

We show that $\deckrip$ is a morphism over partial epistemic models. 
Obviously it satisfies the saturation property by definition. 
To show that $\deckrip$ preserves the indistinguishability, 
suppose $(\preEqClass{Y_1}{X_1}, Y_1)\sim_a^{\ImProd{\cplI}{\actMf{P}}} (\preEqClass{Y_2}{X_2}, Y_2)$.
This implies $X_1\relK{\cplI}{a} X_2$ and $Y_1 \sim_a^P Y_2$, and in particular 
$a\in \coloring^{\cplP} (Y_1 \cap Y_2)$. 
For all $(\preEqClass{Z_1}{X_1}, Z_1)\in \deckrip\bigl((\preEqClass{Y_1}{X_1}, Y_1)\bigr)$ 
and $(\preEqClass{Z_2}{X_2}, Z_2)\in \deckrip\bigl((\preEqClass{Y_2}{X_2}, Y_2)\bigr)$, 
we have $Z_1\relK{T}{\AliveSet{Y_1\cap Y_2}}Z_2$ by the definition of $\deckrip$. 
In particular, $Z_1\relK{T}{a} Z_2$. 
Hence $(\preEqClass{Z_1}{X_1}, Z_1)\relK{\ImProd{\cplI}{\actMf{T}}}{a} (\preEqClass{Z_2}{X_2}, Z_2)$. 
For the preservation of atomic propositions, 
suppose 
$(\preEqClass{Y}{X}, Y)\in \ImProd{\cplI}{\actMf{P}}$, 
$(\preEqClass{Z}{X}, Z)\in \deckrip \bigl((\preEqClass{Y}{X}, Y)\bigr)$, and 
$a\in \AliveSet{Y}$. 
By definition, $X\relK{\cplI}{a}X'$ holds for every $X'\in \preEqClass{Y}{X}$. 
This implies that   
$L^{\ImProd{\cplI}{\actMf{P}}}\bigl((\preEqClass{Y}{X}, Y)\bigr) \cap \AtomProps_a =\labSM^{\cplI} (X)\cap\AtomProps_a$.
Similarly, 
$L^{\ImProd{\cplI}{\actMf{T}}}\bigl((\preEqClass{Z}{X}, Z)\bigr) \cap \AtomProps_a =\labSM^{\cplI} (X) \cap \AtomProps_a$. 
Therefore $L^{\ImProd{\cplI}{\actMf{P}}}\bigl((\preEqClass{Y}{X}, Y)\bigr) \cap \AtomProps_a = 
L^{\ImProd{\cplI}{\actMf{T}}}\bigl((\preEqClass{Z}{X}, Z)\bigr) \cap \AtomProps_a$.

Finally, let us show that $\deckrip$ satisfies the condition~\eqref{eq:kripsolve}. 
Suppose 
$\deckrip\bigl((\preEqClass{Y}{X}, Y)\bigr)=$ $
\satop_{\AliveSet{Y}}\bigl((\preEqClass{Z}{X}$, $Z)\bigr)$
where $Z\in T$, $\dectop(Y) \subseteq Z$, $\AliveSet{Z}\subseteq\AliveSet{X}=\Ag$.
Assume $X'\in \preEqClass{Y}{X}$.
By definition, 
$X \sim_{\AliveSet{Y}}^{\cplI} X'$ and $\cplI, X'\models \precond^P (Y)$. 
The latter implies $Y\in \Psi(X')$ and thus $\dectop(Y)\in \dectop(\Psi(X'))$. 
Since $\dectop(\Psi(X'))\subseteq \Delta(X')$ by~\eqref{eq:topsolve}, 
there exists $Z'\in T$ such that $\dectop(Y)\subseteq Z'$ and $Z'\in \Delta(X')$. 
The former implies $\dectop(Y)\subseteq Z\cap Z'$ and thus $Z \relK{T}{\AliveSet{Y}}Z'$;
The latter implies $\cplI,X'\models \precond^T(Z')$ and hence $\preEqClass{Z'}{X'}$.
By these, we obtain 
$(\preEqClass{Z}{X},Z) \sim_{\AliveSet{Y}}^{\ImProd{\cplI}{\actMf{T}}} 
(\preEqClass{Z'}{X'},Z')$. 
Therefore we have $(\preEqClass{Z'}{X'},Z') \in \deckrip \bigl((\preEqClass{Y}{X}, Y)\bigr)$,
where $X' \in \preEqClass{Z'}{X'}$ trivially holds.

\end{proof}

\begin{theorem}
    Suppose a simplicial task $\anglpair{\cplI, \cplP, \Psi}$ 
    and a simplicial protocol $\anglpair{\cplI, \cplO, \Delta}$ are given, where
    $\cplI=\anglpair{V^{\cplI}, S^{\cplI}, \coloring^{\cplI}, \labSM^{\cplI}}$ is an input simplicial model.
    Let $\actMf{P}=\actionkappa(\anglpair{\cplI, \cplP, \Psi})$ and 
    $\actMf{T}=\actionkappa(\anglpair{\cplI, \cplO, \Delta})$.
    Then 
    a task $\anglpair{\cplI, \cplO, \Delta}$  is solvable by
    a protocol $\anglpair{\cplI, \cplP, \Psi}$ 
    if and only if $\actMf{T}$ is solvable by $\actMf{P}$.
\end{theorem}

\begin{proof}
    Follows from Lemma \ref{lem:kriptosimpldecision} and \ref{lem:smpltokripdecision}.
\end{proof}

\end{document}